\documentclass[11pt]{iopart}

\usepackage{braket,dsfont,color,bm}

\usepackage[bottom=1in,left=1in,right=1in,top=1in]{geometry}
\usepackage{graphicx}   
\usepackage{multirow}
\usepackage[normalem]{ulem}
\usepackage{epsfig,amssymb}
\usepackage{changepage}
\usepackage{bbm}
\usepackage{bm}
\usepackage{times}
\usepackage{amsthm}
\usepackage[caption=false]{subfig}
\usepackage{multirow}
\usepackage[normalem]{ulem}
\usepackage{braket}

\usepackage{dsfont,xcolor}
\usepackage{soul}
\usepackage{etoolbox}

\definecolor{rev}{rgb}{0,0,0}
    
 \newcommand{\dd}{\mathrm{d}}

\newtheorem{theorem}{Theorem}[section]

\newtheorem{corollary}{Corollary}[section]
\newtheorem{lemma}{Lemma}[section]

\makeatletter
\def\@mkboth#1#2{}

\newlength\appendixwidth
\preto\appendix{\addtocontents{toc}{\protect\patchl@section}}
\newcommand{\patchl@section}{%
  \settowidth{\appendixwidth}{\textbf{Appendix }}%
  \addtolength{\appendixwidth}{.1em}%
  \patchcmd{\l@section}{.1em}{\appendixwidth}{}{\ddt}%
}
\makeatother

\newcommand{\binom}[2]{\left( \begin{array}{c} #1 \\ #2 \end{array} \right)}

\usepackage{booktabs}      

\usepackage{cite}

\newcommand{\QFI}{F_{\mathrm Q}}

\begin{document}

\title[]{Precision bounds for frequency estimation under \\collective dephasing and open-loop control}

\author{Francisco Riberi$^{1,2}$,
Gerardo Paz-Silva$^{3,4}$, and Lorenza Viola$^2$}

\address{$^1$Electrical and Computer Engineering Department, University of New Mexico, 498 Terrace St NE, Albuquerque, NM 87106} 

\address{$^2$Department of Physics and Astronomy, Dartmouth College, 6127 Wilder Laboratory, \\ Hanover, NH 03755, USA}

\address{$^3$ Center for Quantum Computation and Communication Technology (Australian Research Council), Centre for Quantum Dynamics, Griffith University, Brisbane, Queensland 4111, Australia}

\address{$^4$ Diraq Pty Ltd, Sydney, NSW, Australia}

\begin{abstract}
Dephasing noise is a ubiquitous source of decoherence in current atomic sensors. We address the problem of entanglement-assisted frequency estimation subject to classical dephasing noise with full spatial correlations (collective) and  arbitrary temporal correlations.  Our contributions are threefold. (i) We derive rigorous, state-independent bounds on the achievable estimation precision, showing how they are entirely determined by the short-time behavior of the decoherence function. For temporally uncorrelated (Markovian) dephasing, precision is limited by a probe-independent constant. For temporally correlated stationary noise, the bound approaches the noiseless limit for classical states, precluding any asymptotic quantum advantage. (ii) We show that these scaling bounds are tight, by constructing generalized Ramsey protocols that saturate them. These optimal protocols use squeezing at the input and before readout, both of which are available in state-of-the-art atomic interferometers. Implementing a perfect-echo protocol, which reaches Heisenberg scaling in the absence of noise, remains optimal in this noisy setting, irrespective of the noise temporal correlations. (iii) We prove that arbitrary collective open-loop control cannot lift the no-go for super-classical precision scaling under either Markovian or colored stationary noise, highlighting the detrimental nature of full spatial correlations. In the latter case, temporal correlations may nonetheless enable constant-factor improvements over the standard quantum limit, which may still be important in practical metrological scenarios.
\end{abstract}


%
%
%

\section{Introduction}

Entanglement-assisted quantum metrology exploits nonclassical correlations in the quantum state of the sensing system to estimate physical parameters with a precision exceeding that achievable by classical strategies \cite{SmerziRMP,RDDrev,BraunRMP}.  Such techniques have found application in state-of-the-art interferometric measurements of physical quantities ranging from frequency to magnetic fields and weak forces, and play a central role  in time-keeping standards \cite{clocks,Ye,Colombo,DegenRMP,SmerziRMP}. 
In all of these settings, precision is fundamentally constrained by both the available quantum resources and unavoidable environmental noise.

A standard interferometric metrology protocol for estimating a parameter of interest, say, $b$, employs $N$ quantum probes prepared in a suitable input state $\rho_0$.  Evolution under the action of a Hamiltonian $H_{\rm S}  \equiv b H_0(t)$, possibly including time-dependent noise, results in a $b$-dependent output state $\rho_b(t)$. Information about $b$ is collected over a total experimental time $T\equiv \nu t$ by repeatedly measuring a system observable over $\nu$ ``shots'', each of duration $t$. 
Then, an estimator $\hat b$ is constructed in a post-processing stage, with precision $\Delta \hat{b}$ lower-bounded in terms of the quantum Fisher information (QFI). For classical, separable probe states,  precision is limited by the {\em standard quantum limit} (SQL), $\Delta \hat b \propto N^{-1/2}$. Entangled inputs can instead surpass this scaling: in the absence of noise, Greenberger-Horne-Zeilinger (GHZ) states are known to saturate the optimal Heisenberg-limited precision (HS) of linear metrological schemes, $\Delta \hat b \propto N^{-1}$; likewise, spin-squeezed states provide a robust experimental route towards achieving HS, albeit with a worse prefactor \cite{Escher,Guta,SmerziRMP,DaveSqueezing1,DaveSqueezing2}. 

In realistic settings, however, the probes couple to uncontrolled external degrees of freedom. This coupling induces decoherence, that can severely degrade the metrological performance. When the noise acts transversely to the signal Hamiltonian, active mitigation strategies based on dynamical decoupling (DD) and quantum error correction can, in principle, restore super-classical scaling \cite{Chaves2,Sekatski,DurDD, Layden2,Zhou,Mann}. By contrast, when the noise operator commutes with the signal Hamiltonian, as it is the case for dephasing noise, mitigation strategies are restricted. This type of ``parallel'' fluctuations are a dominant source of decoherence in atomic and solid-state platforms \cite{Khor2009,BECCollDeph,Dorner2012,FelixPRA}. Entangled states are particularly affected, as the rapid decay of coherences strongly suppresses their usefulness. The extent of this degradation depends sensitively upon the spatial and temporal correlations of the underlying noise.

When the noise is fully uncorrelated in space and time, improvements over the SQL are precluded beyond a constant factor \cite{Huelga1997,RDDelusive,RDDefficient}. Spatially local but temporally correlated stationary noise is known to be fundamentally bounded by an intermediate \emph{Zeno regime}, $\Delta \hat{b}\propto N^{-3/4}$ \cite{Macies2015,HaaseNJP}.  
For \emph{partial} spatial correlations, several strategies have been proposed to retain a quantum advantage, ranging from randomized metrology to noise-aware state engineering and sensor's geometry optimization \cite{Jeske2014, FUR, Bias}. 

The fully collective dephasing regime, where all probes are identically coupled to the fluctuations, is harder to counter: spatial noise-tailoring is ineffective, and quantum error correction offers no protection as the signal Hamiltonian lies in the span of the Lindblad operators \cite{Layden2,Sekatski}. Likewise, while decoherence-free subspaces provide in principle a pathway to avoid this kind of collective noise \cite{Dorner2012}, they tend to be fragile to perturbations and experimentally difficult to scale to large system size. 

Time-dependent open-loop control is another resource that has received some attention in this scenario. Both pulsed DD and continuous driving with collective spin components have been found to improve single-shot QFI performance for low-dimensional systems \cite{Haidong2017,zhai2023control2,NMcontrol}. To our knowledge, however, no general analysis has been carried out in the asymptotic regime where $N \gg 1$. \emph{A priori}, it would seem reasonable to expect that suitable DD techniques could help mitigating temporally correlated dephasing. This intuition is rooted in the fact that, despite sharing a quantization axis, the control-induced filtering could potentially modulate signal and noise in different ways -- ultimately enabling the noise to be suppressed by retaining the target signal to {\em some} extent. In the current state of the literature, it remains unclear whether sub-SQL precision scaling is fundamentally precluded, both assisted by DD or otherwise.

Our work aims to fill this gap, by addressing the above questions and deriving rigorous precision bounds for frequency estimation subject to classical collective dephasing. Our bounds are valid for any choice of input states and measurement basis.  In the absence of applied control, our analysis builds on existing results for phase estimation \cite{Escher,Escher20122}, and extends them by explicitly incorporating optimization over the interrogation time, which is crucial in our setting. Importantly, we show that the achievable precision scaling depends entirely on the short-time behavior of the decoherence function or, equivalently, on the temporal-correlation structure of the noise. For temporally uncorrelated (Markovian) dephasing with a white spectrum, as well as for stationary noise with a colored spectrum, our bounds preclude any asymptotic improvement beyond SQL scaling. For non-stationary colored noise, super-classical scaling is not ruled out, but remains fundamentally constrained.

Beyond establishing precision limits, we construct interferometric protocols based on two input classes that saturate these bounds. GHZ states achieve the optimal scaling up to a constant factor, but remain challenging to implement experimentally. Spin-squeezed states generated by one-axis-twisting, so-called one-axis-twisted states (OATS) \cite{Kita1993, OATS2023}, are robust to imperfect preparation and capable of attaining the optimal precision through the application of a disentangling operation before readout. Remarkably, we show that the best protocol in the presence of collective dephasing coincides with the noiseless optimal strategy \cite{Monika}, establishing a direct and robust pathway to implementation in realistic platforms.

Finally, we examine the potential of open-loop, collective Hamiltonian control as a resource to increase precision in this setting. We establish two general {\em no-go theorems}, showing that superclassical precision scaling remains precluded under control,  although constant-factor improvements are possible for temporally correlated noise, in a sense we make precise.

The content is organized as follows. Section\,\ref{Background} provides the necessary quantum estimation and open-quantum system background, by introducing the model and relevant noise assumptions. In Sec.\,\ref{general}, we derive state-independent precision bounds and analyze the role of temporal correlations determining the asymptotic precision scaling. We discuss the performance of protocols using GHZ and spin-squeezed input states in Sec.\,\ref{saturating}. The benefits of open-loop 
control as a metrological resource is discussed in Sec.\,\ref{openloop}, with special emphasis on schemes employing sequences of arbitrary instantaneous control pulses. We close in Sec.\,\ref{conclusions} by providing a summary of our results, along with a discussion of some outstanding open questions. Additional discussion on noise properties and technical derivations are included in three separate appendices.

\section{Background}
\label{Background}

\subsection{Frequency estimation by Ramsey interferometry} 
\label{int}

We consider a single-parameter frequency estimation problem within a linear metrological setting, whereby a sensing system comprising $N$ qubit probes is prepared in an initial state $\rho_0$ and evolves for an interrogation time $t$ under dynamics that encode the target signal $b$, without additional interaction between probes being permitted (see \cite{APLQ} for a discussion of nonlinear noisy frequency estimation). If the resulting output state is $\rho_b(t)$, the measurement of a system observable $O=\sum_i \mu_i \Pi_i$ yields outcome $\mu_i$ with probability $P(\mu_i)=\mathrm{Tr}[\rho_b(t)\Pi_i]$. Repeating the protocol $\nu$ times  
for a total runtime $T = \nu t$ 
results in the data string  $\vec{\mu}=(\mu_1, \ldots \mu_{\nu})$, which contains information about $b$.  Then, a suitably chosen function $\hat{b}(\vec{\mu})$ assigns an estimate $\hat b$ of the parameter. When, in a ``local estimation setting'',  prior knowledge restricts $b$ to a neighborhood of a known value $b_0$, we can compute the precision by the so-called {\em method of moments} \cite{Pezze20142,SmerziRMP,Bias}:
\begin{eqnarray}
 \qquad \quad \Delta \hat{b}^2(t)  
 =  \frac{\Delta_{\rho_b}^2 O(t) }{
 \nu \big[\partial_b \langle O(t) \rangle {\vert_{b_0}} \big]^{2}} = \frac{t\,\Delta^2_{\rho_b} O(t) }{
 T \big[\partial_b \langle O(t) \rangle {\vert_{b_0}} \big]^{2}}, 
\label{mom} 
\end{eqnarray}
with $\Delta_{\rho_b}^2 O(t)\equiv [\langle O^2(t)\rangle - \langle O(t)\rangle^2 ]$ denoting the variance taken with respect to $\rho_b(t)$. Throughout this work, we treat the total number of probes, $N$, as a parametric variable, and 
we adopt an idealized 
model in which each measurement shot of duration $t$
is repeated arbitrarily many times within a \emph{fixed} total experimental time $T$, with no additional overheads \footnote{In this way, we neglect dead time associated with state preparation, control pulses, readout, or classical post-processing, as well as finite control bandwidth, pulse-duration constraints, or hardware-imposed minimum cycle times. This idealization allows us to isolate {\em fundamental} precision limits arising from collective dephasing alone and to perform an {\em unconstrained} optimization over $t$. While such overheads are inevitably present in realistic implementations, they are highly platform-dependent and would obscure the universal scaling behavior that is the focus of the present analysis.}. All bounds and optimality statements are formulated under these resource constraints.


Equation (\ref{mom}) provides a convenient operational benchmark; however, it clearly depends on the choice of the operator $O$. The ultimate limit to performance for a given output $\rho_b(t)$ entails an optimization over all possible measurement strategies. For an {\em unbiased} estimator $\hat b$ (see \cite{Bias} for a discussion of noise-induced bias effects), the achievable precision is lower-bounded by the {\em quantum Cram\'{e}r-Rao bound} (QCRB): 
\begin{equation}
\qquad \Delta \hat b^2 (t) \geq \Delta \hat b^2_{\rm QCRB}(t) 
=\frac{1}{\nu \, \QFI[\rho_b(t)]}= \frac{t}{T  \QFI[\rho_b(t)]}. 
\label{eq:QCRB}
\end{equation}
Here, $\QFI[\rho_b(t)]$ denotes the QFI of the output state with respect to $b$. The QFI  quantifies the maximal amount of information we can extract about the parameter $b$, given $\rho_b(t)$. Formally, it can be expressed in terms of the fidelity of two states with infinitesimally differing parameters \cite{NielsenChuang, Hubner2}, 
\begin{equation}
F_Q[\rho_b]\, db^2 = 8 (1- \mathcal{F}(\rho_b,\rho_{b+\delta b})), \qquad \mathcal{F}(\rho, \sigma) \equiv {\rm Tr} \Big[\sqrt{\sqrt{\rho} \,\sigma\, \sqrt{\rho}} \,\Big].
\label{QFIFid3}
\end{equation}
Equivalently, it is given by the variance of a zero-mean, Hermitian operator, the so-called symmetric logarithmic derivative (SLD) $L$, which is implicitly defined through the equation
$F_Q[\rho_{b}] = \Tr\!\left[\rho_b L^2\right].$
When the eigendecomposition of the state $\rho_b \equiv  \sum_i \lambda_i \, |i\rangle\!\langle i| 
$ is known, one may evaluate both QFI and SLD explicitly \cite{Pezze20142}:
\begin{equation}
L =
2 \sum_{\lambda_i+\lambda_j>0}
\frac{
\langle i | \partial_b \rho_b | j \rangle
}{
\lambda_i + \lambda_j
}
\, |i\rangle\!\langle j|,
\qquad
F_Q[\rho_b]
=
2 \sum_{\lambda_i+\lambda_j>0}
\frac{
\left| \langle i | \partial_b \rho_b | j \rangle \right|^2
}{
\lambda_i + \lambda_j
}.
\label{eq:QFIeigen}
\end{equation}
The QCRB can be saturated through error propagation,  Eq.\,(\ref{mom}),  by measuring an observable $O$ which is diagonal in the eigenbasis of the SLD operator $L$. However, in noisy settings, it is typically not possible to diagonalize $\rho_b(t)$ exactly, and more sophisticated methods to assess metrological performance are required.

\subsection{Classical collective dephasing}
\label{sett}

We assume that the $N$-qubit sensing system is subject to classical collective dephasing. In the absence of external control, the dynamics is then governed by a noisy Hamiltonian of the form 
\begin{equation}
   \qquad \qquad  \qquad H_0(t)= J_z\,[b+ \xi(t)], 
    \label{Ham}
\end{equation}
where $J_z$ denote the $z$-component of the collective spin angular momentum, and $\xi(t)$ is a classical stochastic process describing phase noise. Specifically, we assume the fluctuations to be zero-mean and Gaussian, implying they are fully characterized by their first two cumulants:
\begin{equation*}
   \qquad \quad \langle \xi(t)\rangle=0, \;\forall t, \qquad\langle \xi(t_1) \xi(t_2)\rangle=C(t_1,t_2), 
\end{equation*}
where $\langle \cdot\rangle$ denotes the expectation over noise realizations and $C(t_1,t_2)$ is the two-point (auto)correlation function, which describes the noise temporal structure.
The evolution under the Hamiltonian in Eq.\,(\ref{Ham}) may be exactly represented as a Gaussian average over stochastic unitary trajectories \cite{Durkin2,DurDD}, namely, 
\begin{equation}
\begin{array}{rcl}
\rho_b(t) &=& \displaystyle
\int_{-\infty}^{\infty} \dd \lambda (t)\;
\mathcal{N}_{\lambda}(0,\chi(t))\;
e^{-i (b t + \lambda(t)) J_z}\,
\rho_0\,
e^{i (b t + \lambda(t)) J_z}
\\[1ex]
&& \displaystyle
\mathcal{N}_{\lambda}(0,\chi(t))
\equiv \frac{1}{\sqrt{\pi \chi(t)}}\,
e^{-\lambda(t)^2/\chi(t)}
\end{array}
\label{eq:stochastic_map}
\end{equation}
where $\lambda(t) \equiv \int_{0}^t\!\dd s \,\xi(s)$ is a classical (still zero-mean and Gaussian) stochastic process that quantifies the random phase accumulation, and $\chi(t)$ is a decay function that describes the influence of the noise on the probe dynamics. The role of $\chi(t)$ in modulating the suppression of coherence elements is  seen more transparently by rewriting Eq.\,(\ref{eq:stochastic_map}) in the collective (Dicke) spin basis of $J^2, J_z$ eigenstates:
\begin{equation}
   \qquad \qquad \rho_b(t)= \!\!\sum_{m,m'=-J}^J e^{-i b t (m-m')}\,e^{-\chi(t) (m-m')^2}\,|J, m\rangle \langle J,m'|, 
    \label{eq:Dicke}
\end{equation}
which will be convenient for establishing analytical bounds and constructing optimal probe states.

\subsection{Noise temporal correlations and short-time coherence decay}

The decay function entering Eq.\,(\ref{eq:Dicke}) may be expressed as the double integral of the two-point correlation function,
\begin{equation}
\qquad \qquad\chi(t)= \langle \lambda^2 (t) \rangle\equiv \int_0^t \!\! \dd s \int_0^t \!\! \dd s' \, C(s,s').
\label{chit}
\end{equation}
Generally, we may assume that $\chi(t)$ admits a short-time power-law expansion of the form \cite{Degen}:
\begin{equation}
\qquad \quad \chi(t) \simeq \chi_0^n \,(\omega_c t)^{n}, \qquad n \geq 1,\;  \omega_c t \ll 1, 
\label{eq:chi_powerlaw} 
\end{equation}
where $\chi_0^n > 0$ is a dimensionless constant and $\omega_c$ a parameter with units of frequency. The exponent 
$n$ may be taken to characteriz the short-time ``softness'' of the noise and directly determines the attainable metrological scaling, as we shall show in Secs.\,\ref{general} and \ref{openloop}. 

Important distinctions arise depending on whether the noise is stationary or not. When the noise is (wide-sense) stationary, $C(t_1,t_2)$ depends only on the time lag $\tau\equiv |t_2-t_1|$, and $C(0)$ coincides with the noise variance. For stationary noise, we can write Eq.\,(\ref{chit}) in the frequency domain in terms of the power spectral density (simply ``spectrum'' henceforth) $S(\omega)$ -- the Fourier transform of $C(\tau)$ -- and a first-order ``filter-function'' $F(t,\omega)$ \cite{Paz2014, FelixPRA, FUR}, that is, 
\begin{equation}
\chi(t) =\int_{- \infty}^{\infty} \dd \omega\, \frac{S(\omega)}{2 \pi}\, 
F(t, \omega), \qquad 
S(\omega)\equiv  \int_{-\infty}^{\infty}\!\!\dd\tau\,e^{-i \omega \tau} C(\tau), 
\label{chit2}
\end{equation}
where, for the case of free evolution, we simply have 
\begin{equation}
\qquad \qquad F(t,\omega)\equiv \bigg\vert \int_0^t \!\dd s \, e^{i \omega s} \bigg\vert^2 = 
\frac{2\sin(\omega  t/2)}{\omega}. 
\label{chit3}
\end{equation}
As one may show (see Appendix \ref{app:formal_spectrum}), the requirements that $C(\tau)$ be continuous at $\tau=0$ and the variance be finite, $C(0)<\infty$, are {\em sufficient} conditions for the short-time decay behavior to be {\em quadratic} ($n=2$), 
$$ \chi(t) \simeq\chi_0^2 (\omega_c t)^2.$$ 
We can interpret $\omega_c$ as a cutoff frequency setting the characteristic timescale of the noise correlations $\tau_{\rm B} = \omega_c^{-1}$. This case comprises most of the temporally correlated dephasing process typically encountered in physically relevant settings, including Ornstein-Uhlenbeck noise \cite{Walsh} (for which the autocorrelation is not analytic at zero), or even $1/f^\alpha$ noise processes, ubiquitously arising in the solid state \cite{faoro,paladino}.

Non-quadratic short-time behavior is nonetheless possible when the above requirements do not apply. Notably, Gaussian white noise is still stationary in a generalized sense, but with a singular correlator, $C(\tau)= \chi_0 \omega_c \,\delta(\tau)$, with $\delta$ denoting the Dirac's distribution. In this case, which physically corresponds to temporally uncorrelated, Markovian noise, Eq.\,(\ref{chit}) leads to an exact {\em linear} behavior: $\chi(t)= \chi_0\, \omega_c t $, with $n=1$.

Higher-order exponents, $n>2$, arise only (although not necessarily) when stationarity is violated, so that $C(t_1,t_2)$ is not solely a function of the time lag. Then, the spectrum $S(\omega)$ does not exist in the Wiener-Khinchin sense \cite{MOCbook}, and $\chi(t)$ does not admit the simple representation given in Eq.\,(\ref{chit2}). For instance, the decay corresponding to Brownian motion, relevant for instance to describe diffusion processes in nuclear magnetic resonance \cite{Chakrabarti_2016}, is cubic in time: 
 \begin{equation*}
     \!\!\!\!\!\!\!\!\!\!\!\!\!\!\!\! \chi(t)= \int_0^t\!\!\dd s\int_0^t\!\!\dd s' \langle \xi(s) \xi(s')\rangle= \chi_0 \omega_c^3\int_0^t\!\!\dd s\int_0^t\!\!\dd s' \min(s,s')= \frac{2}{3} \chi_0 (\omega_c t)^3.
 \end{equation*}
When $\xi(t)$ is the integrated version of a stationary process -- say,  $\xi(t)\equiv \omega_c\int_0^t\! \dd s\, \eta(s)$, with $\eta$ (wide-sense) stationary -- we may still write $\chi(t)$ in frequency domain (see details in Appendix \ref{app:formal_spectrum}):
\begin{equation*}
 \!\!\!\!\!\!\!\!\!\!\!\!\!\!\!\!   \chi(t)\equiv  \int_{-\infty}^{\infty}\frac{\dd\omega}{2\pi}\, \frac{S_\eta(\omega)}{\omega^2} \,
 G (t,\omega), 
\quad G(t, \omega)\equiv  \left| \int_{0}^t\!\!\dd s\int_0^s\!\!\dd u \,e^{i \omega u} \right|^2 = 
\frac{|1-e^{i\omega t}+i\omega t|^2}{\omega^2},
\end{equation*}
where $S_\eta(\omega)/\omega^2$ can be {\em formally} identified as $S_\xi(\omega)$. Importantly, however, the relevant filter function $G(t, \omega)$ is now effectively a second-order filter, thus preventing any direct connection between the short-time behavior of $\chi(t)$ and the high-frequency decay of $S_\xi(\omega)$.

\section{Fundamental precision bounds under collective dephasing}
\label{general}

In this section we derive state-independent bounds on the precision of frequency estimation under collective dephasing. The bounds apply to arbitrary initial probe states and choice of measurement basis, explicitly incorporating optimization over the interrogation time. 

\subsection{QFI growth under collective dephasing}
\label{QFIg}

For an interrogation time $t$, the ensemble-averaged probe state under collective dephasing can be expressed in the $J_z$ eigenbasis by using Eq.\,(\ref{eq:Dicke}). Physically, the decay function $\chi(t)$ constrains the attainable QFI by suppressing the distinguishability between output states at nearby values of the frequency $b$. Following the variational method of Escher \emph{et al.}~\cite{Escher}, we introduce a fictitious environment $\rm B$ and consider a purification $\ket{\Phi_b(t)}\in\mathcal{H}_{\rm SB}$ satisfying $\Tr_{\rm B}\!\left[\ket{\Phi_b(t)}\!\bra{\Phi_b(t)}\right]=\rho_b(t)$. This approach is particularly natural in the present setting, since collective dephasing can be described in terms of a random-unitary channel, Eq.\,(\ref{eq:stochastic_map}) see also \cite{FUR} for an explicit realization in terms of a bosonic environment). 
As measurements on the system are a subset of measurements on the enlarged space, the QFI satisfies
\begin{equation}\qquad \qquad \quad
\QFI\!\left[\rho_b(t)\right]\le \QFI\!\left[\ket{\Phi_b(t)}\!\bra{\Phi_b(t)}\right].
\label{eq:purif_ub}
\end{equation}
The bound can be tightened by exploiting the non-uniqueness of this construction: any other purification is of the form
$\ket{\Psi_b(t)}=u_{\rm B}(b)\ket{\Phi_b(t)}$, for some unitary $u_{\rm B}(b)$ acting only on $\mathcal{H}_B$~\cite{Escher}.
Defining
\begin{equation}
\qquad  \quad  H_{\rm SB}(b)\equiv  i\,\partial_b \ket{\Phi_b(t)}\!\bra{\Phi_b(t)}, 
\qquad 
h_{\rm B}\equiv  i\,(\partial_b u_{\rm B}^\dagger) u_{\rm B},
\end{equation}
Escher \emph{et al.} show that the QFI of the purified family can be expressed as a variance with respect to $\Phi_b(t)$:
\begin{equation}
\qquad \quad \QFI\!\left[\ket{\Psi_b(t)} \bra{\Psi_b(t)}\right]=4\,\Delta^2_{\Phi_b(t)} H, \qquad H \equiv H_{\rm SB}-h_{\rm B} .
\label{eq:Escher_variance}
\end{equation}

We apply this method to minimize Eq.\,(\ref{eq:Escher_variance}) for a concrete purification and a one-parameter family of bath unitaries. A convenient choice producing Eq.\,\ref{eq:Dicke} is obtained by appending a single harmonic oscillator bath mode initially in the vacuum and applying a $J_z$-dependent displacement,
\begin{equation}
\ket{\Phi_b(t)}
= e^{-i b t J_z}\,
\ket{\vartheta}\otimes
e^{\sqrt{\chi(t)}\,(a^\dagger-a)\,J_z}\ket{0}_{\rm B} ,
\label{eq:purification_choice}
\end{equation}
where $\rho_0=\ket{\vartheta}\!\bra{\vartheta}$ is the input state, and $\{a,a^\dagger\}$ are bosonic operators.
For this purification, $H_{\rm SB}=t\,J_z$ acts on the system only. Consider next the bath-local  family of ``gauge transformations'' depending on the free parameter $\zeta$ 
\begin{equation}
u_B(b)=\exp\!\bigg[\!-i\, b \zeta\, t\,\frac{a^\dagger+a}{2\sqrt{\chi(t)}}\bigg],
\qquad 
h_B(\zeta)=\zeta\,t\,\frac{a^\dagger+a}{2\sqrt{\chi(t)}}.
\label{eq:gauge_family}
\end{equation}
Evaluating Eq.\,(\ref{eq:Escher_variance}) on the state (\ref{eq:purification_choice}), we obtain (see Appendix \ref{varcalc} for details)
\begin{equation}
\qquad \QFI\!\big[\ket{\Psi_b(t)} \bra{\Psi_b(t)}\big] =
4t^2 \biggl[(1-\zeta)^2\,\Delta_{\rho_b}^2 J_z+\frac{\zeta^2}{4\,\chi(t)}\bigg].
\label{eq:QFI_lambda}
\end{equation}
Minimizing over $\zeta$, we find
\begin{equation}
\qquad \qquad \qquad \zeta_{\rm opt}=\frac{4\,\chi(t)\,\Delta_{\rho_b}^2 J_z}{1+4\,\chi(t)\,\Delta_{\rho_b}^2 J_z}.
\end{equation}
Therefore, the variationally optimized purification yields the following state-independent bound for the single-shot QFI:
\begin{equation}
\qquad \qquad \QFI\!\left[\rho_b(t)\right] \le
\frac{4t^2\,\Delta_{\rho_b}^2 J_z}{1+4\,\chi(t)\,\Delta_{\rho_b}^2 J_z} \leq \
\frac{N^2 t^2}{1+\chi(t)\,N^2},
\label{eq:QFI_state_indep_purif}
\end{equation}
where the second inequality follows from  $\Delta_{\rho_b}^2 J_z\le \|J_z\|^2=(N/2)^2$, valid for any initial state.
Given a fixed runtime $T=\nu t$, the corresponding bound on the total QFI, obtained by multiplying by $\nu$, satisfies
\begin{equation}
\qquad \quad \qquad F_{Q\,\rm tot}(N,T)\le 
T\,\max_{t>0}\;
\frac{N^2 t}{1+\chi(t)\,N^2}.
\label{eq:QFI_tot_bound}
\end{equation}

\subsection{Time optimization and asymptotic precision bounds}

We now optimize Eq.\,(\ref{eq:QFI_tot_bound}) with respect to interrogation time by taking the decay coefficient to obey the short-time noise law (\ref{eq:chi_powerlaw}). We distinguish between Markovian and colored fluctuations:

\paragraph{Markovian noise ($n=1$).}
Here, $F_{Q\,\rm tot}(N,T)$ is a strictly increasing function of $t$ with no finite maximizer. The optimized total-QFI bound follows:
\begin{equation}
F_{Q\,\rm tot}(N,T) \leq T \sup_{t>0} \frac{N^2 t}{1+N^2\chi_0\omega_c t}=T\lim_{t\to\infty}\frac{N^2 t}{1+N^2\chi_0\omega_c t}
=\frac{T}{\chi_0\,\omega_c},
\label{eq:QFI_tot_opt_markov}
\end{equation}
with the understanding that any experimental upper bound on $t$ -- or additional constraints not captured by Eq.\,(\ref{Ham}) -- may fix the interrogation time in practice.

\paragraph{Colored noise ($n\geq2$).}
In this case, solving Eq.\,(\ref{eq:QFI_tot_bound}) for $\chi(t) \simeq \chi_0^n (\omega_c t)^n$ yields a unique stationary point $t^{\star}$. Evaluating at $t^{\star}$ produces an exact time-optimized bound. The performance is determined by the following expressions:
\begin{equation}
\!\!\!\!\!\!\!\!\!\!t^\star= \frac{1}{ \chi_0 \,\omega_c} \left(\frac{1}{n-1} \right)^{1/n} N^{-2/n}, \quad F_{Q\,\rm tot}(N,T)\le
\frac{T}{\chi_0\,\omega_c}\;
\frac{(n-1)^{\,1-\frac{1}{n}}}{n}\;
N^{\,2-\frac{2}{n}}.
\label{eq:QFI_tot_opt_exact}
\end{equation}


Using the QCRB, Eq.\,(\ref{eq:QCRB}), state-independent lower-bounds to the precision follow:
\begin{equation}
\Delta \hat{b}_{\rm tot}^2
\geq
\left\{
\begin{array}{ll}
\frac{ \chi_0 \,\omega_c}{T}, & n=1, \\[1.2ex]
\frac{\chi_0\,\omega_c}{T}\,
g(n)\,
N^{-\frac{2(n-1)}{n}}, & n>1,
\end{array}
\right.
\qquad
g(n)\equiv n (n-1)^{-1+\frac{1}{n}}.
\label{eq:QFI_summary_bounds}
\end{equation}
We can draw several conclusions from Eqs.(\ref{eq:QFI_tot_opt_markov})-(\ref{eq:QFI_tot_opt_exact}).
First, in the Markovian case $n=1$, precision cannot increase with the number of probes $N$.
Second, for colored noise, the bound allows for super-classical scaling with $N$ only for $n>2$. 
This establishes the SQL as a fundamental limitation for frequency estimation under collective stationary 
dephasing with a colored spectrum.
Third, even when $n>2$, HS is excluded unless the noise is effectively absent in the short-time limit.
%
Taken together, these results show that the possibility of surpassing the SQL is determined \emph{ entirely} by the short-time structure of the dephasing noise. Since the bounds derived here are independent of the probe state and measurement,  they constitute fundamental precision limits. In the following section, we examine concrete families of input states and corresponding measurement strategies, showing how these bounds can be saturated in practice.

\section{Saturating the precision bounds with optimal 
protocols}
\label{saturating}


\subsection{Interferometry with GHZ states}
\label{GHZstates}

We consider the $N$-qubit GHZ state, written in terms of Dicke states:
\begin{equation}
\qquad \qquad \ket{\mathrm{GHZ}}=\frac{1}{\sqrt{2}}\Bigl(\ket{J,J}+\ket{J,-J}\Bigr), \qquad J= N/2.
\label{eq:GHZ_def}
\end{equation}
Under the noisy dynamics generated by Eq.\,(\ref{Ham}), the evolved state at interrogation time $t$ remains effectively two dimensional. We may write it in the 
basis $\{\ket{J,J},\ket{J,-J}\}$ as 
\begin{equation}
\qquad \quad \rho_b^{\rm GHZ}(t)
=
\frac{1}{2}
\left(
\begin{array}{cc}
1 & e^{-i N b t}\,e^{-N^2\chi(t)} \\
e^{i N b t}\,e^{-N^2\chi(t)} & 1
\end{array}
\right).
\label{eq:GHZ_density}
\end{equation}
Thus, collective dephasing  suppresses the off-diagonal coherence at a rate proportional to $N^2\chi(t)$. Because the state in Eq.\,(\ref{eq:GHZ_density}) can be readily diagonalized, we can evaluate the SLD and the total QFI through Eq.\,(\ref{eq:QFIeigen}):
\begin{equation}
\!\!\!\!\!\!\!\!\!\!\!L = e^{- N^2 \chi(t)} \,N\,t\,
\left(
\begin{array}{cc}
0 & i \\
-i & 0
\end{array}
\right),
\qquad
F_{Q\,\rm tot}^{\rm GHZ}(t)
= \frac{T}{t} F_{Q}^{\rm GHZ}(t)
= T\,N^2 t\,e^{-2N^2\chi(t)}.
\label{eq:QFI_GHZ_short}
\end{equation}
Equation (\ref{eq:QFI_GHZ_short}) captures the competition between  phase accumulation, scaling as $N^2 t^2$, and decoherence, which suppresses the signal exponentially in $N^2\chi(t)$. Assuming a power-law behavior for the decay, Eq.\,(\ref{eq:chi_powerlaw}),
we maximize $F_{Q\,\rm tot}^{\rm GHZ}(t)$ with respect to $t$.  This yields an optimal interrogation time and overall upper bound:
\begin{equation}
\!\!\!\!\!\!\!\!\!\!\!\!\!\!\!t_{\rm GHZ}^{\star}
   = \frac{1}{\chi_0 \omega_c}
\left(\frac{1}{2 n\,\,}\right)^{\!1/n} N^{-2/n}, \quad  F_{Q\,\rm tot}^{\rm  GHZ}(N,T)
=
\frac{T}{\chi_0\,\omega_c}\,
e^{-1/n} (2n)^{-1/n}\,
N^{\,2-\frac{2}{n}}.
\label{eq:QFI_GHZ_tot_opt}
\end{equation}
 
\begin{figure*}[t!]
\centering
\includegraphics[width=16.5cm]{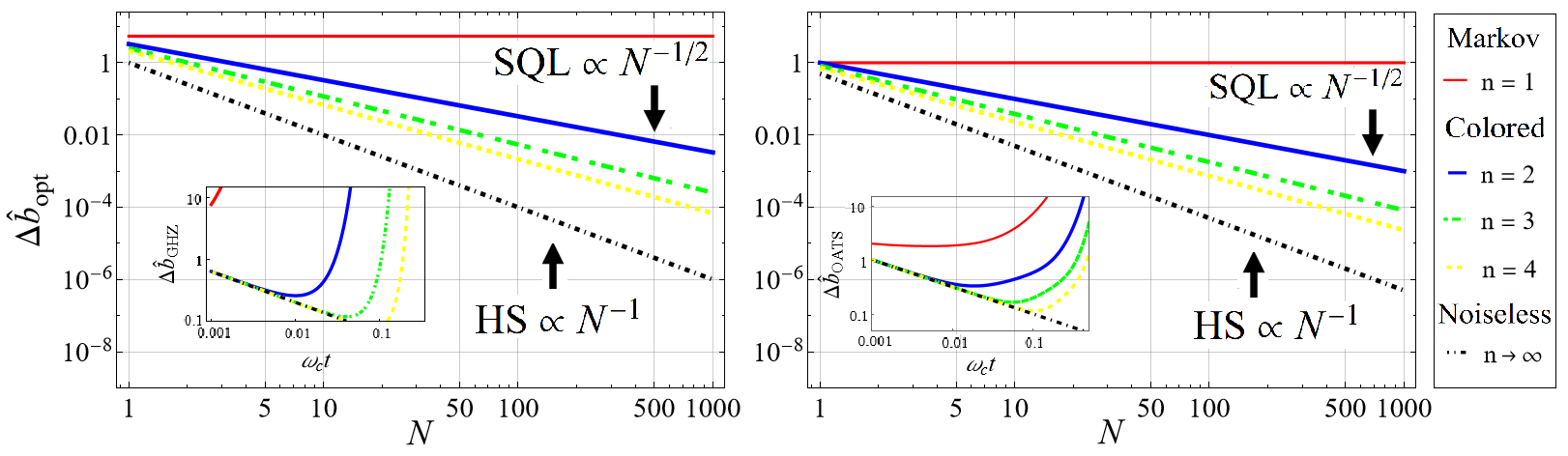} 
 \vspace*{-5mm}
\caption{{\bf Performance of GHZ and spin-squeezed inputs subject to collective dephasing} 
Left: Optimal precision of a GHZ state, Eq.\,(\ref{dbGHZ}), as function of probe number for different short-time behavior of the decay coefficient, $\chi(t)=\chi_0^n (\omega_c t)^n$. The precision saturates the scaling bound in Eq.\,(\ref{eq:QFI_summary_bounds}), including $N$-independent precision for Markovian noise (solid, red), SQL bound for colored, stationary noise (solid, blue), with HS achieved in the noiseless case (black, dot-dashed). Inset: Uncertainty $\Delta \hat{b}(t)$ as a function of time for fixed system size $N=50$. 
Right: Same plot and inset for the optimal perfect echo protocol using OATS, with input squeezing and rotation strength as in \cite{Monika}, see Eq.\,(\ref{eq:QFI_OATS_opt}) and Table\,\ref{table:0}.
Besides being less sensitive to preparation errors, OATS outperform GHZ inputs by a constant factor.  
Parameters:  $\omega_c=1$, $\chi_0^2=1$.}
\label{GHZplot} 
\end{figure*}

\noindent
Replacing into the QCRB, Eq.\,(\ref{eq:QCRB}), the optimal precision follows:
\begin{equation}
   \quad \qquad \qquad  \Delta \hat{b}_{\rm opt}^{\rm GHZ}\geq \sqrt{\frac{\chi_0\omega_c}{T}} \,e^{1/2n} (2n)^{1/2n}\,N^{ -\frac{n-1}{n}}. \label{dbGHZ}
\end{equation}

Comparing Eq.\,(\ref{eq:QFI_GHZ_tot_opt}) with the state-independent bound (\ref{eq:QFI_summary_bounds}), we find that GHZ states reach the fundamental scaling imposed by collective dephasing, up to a noise-dependent constant prefactor. A measurement strategy saturating the QCRB consists in detecting an observable diagonal in the  $L$ eigenbasis, such as the parity operator \cite{SmerziRMP}.
Despite such a favorable scaling, GHZ states concentrate all the metrological sensitivity into a single macroscopic coherence. This makes them particularly susceptible to imperfect preparation and limits their applicability for large probe ensembles \cite{SmerziRMP}.

\subsection{Interferometry with one-axis twisted spin-squeezed states}
\label{OATSstates}

Unlike GHZ states,  spin-squeezed states distribute phase information over many collective degrees of freedom, which improves robustness against implementation errors \cite{SchulteEchoes}. In the limit of an asymptotically large number of probes we consider, $N\rightarrow\infty$, sufficiently polarized  squeezed states behave as a single bosonic mode. The noisy dynamics become effectively Gaussian in phase space, allowing 
for an analytic derivation of the QFI and optimal measurement strategy. Our approach entails two key steps \cite{APLQ}:
(i) Map the discrete spin system to an effective continuous-variable model using the Holstein-Primakoff (HP) transformation.
(ii) Compute the QFI exactly within the resulting Gaussian framework.

\subsubsection{Holstein-Primakoff description of OATS inputs.}
For convenience, we rotate the frame so that both the signal and the noise couple to $J_x$ rather than $J_z$:
\begin{equation}
\qquad \quad H(t) = (b + \xi(t)) J_x, \qquad \rho_0 \mapsto R_y(\pi/2)\rho_0 R_y^\dagger(\pi/2).
\label{eq:H_x}
\end{equation}
This yields an equivalent metrological setting to that of Eq.\,(\ref{Ham}).
We consider initial states belonging to the one-axis twisted family
\begin{equation}
\qquad \qquad \ket{\mathrm{OATS}}_z
= R_z(\beta)\, T_\mu \ket{\mathrm{CSS}}_z,
\qquad
T_\mu = e^{-i\mu J_x^2},
\label{eq:OATS_def}
\end{equation}
where $\ket{\mathrm{CSS}}_z$ is a coherent spin state polarized along $\hat{z}$,
$\mu$ denotes the squeezing strength, and $\beta$ rotates the squeezing axis.
We introduce the bosonic creation and annihilation operators $\{a^{\dagger}, a\}$ ($[a^{\dagger},a]=1$), with associated dimensionless position and momentum operators, 
\begin{equation}
\qquad \qquad \hat{x}\equiv \frac{1}{\sqrt{2}}(a+a^{\dagger}),\qquad\hat{p}\equiv \frac{i}{\sqrt{2}}(a^{\dagger}-a) .   
\end{equation}
Then, we rewrite the collective spin operators in bosonic language by means of the HP transformation \cite{SmerziRMP,Durkin2,APLQ}: 
\begin{equation}
\qquad \qquad \hat{J}_x \mapsto \sqrt{J}\,\hat{x}, \quad \hat{J}_y\mapsto \sqrt{J}\,\hat{p}, \quad \hat{J}_z \mapsto  J-a^{\dagger}a    \label{HP},
\end{equation}
 The mapping in Eq.\,(\ref{HP}) is not exact: it only provides a faithful description of the dynamics if a \emph{low excitation condition} is obeyed at all times,  $\langle a^{\dagger}a (t)\rangle \ll 2J$. Physically, the state must be sufficiently localized to approximate the evolution as occurring in a plane tangent to the North pole of the Bloch sphere. This sets an upper bound to the strength $\mu$ entering 
 Eq.\,(\ref{eq:OATS_def}): $0 \leq |\mu| \lesssim (2J)^{-1/2}$ \cite{SmerziRMP}. In the low-excitation regime,  the dynamics are then ruled by an effective bosonic Hamiltonian: 
\begin{equation}
   \qquad \qquad H_{\rm S}(t) \mapsto H_{\rm HP}(t)=\sqrt{J}\,\hat{x}\,[b+\xi(t)],
    \label{HHP}
\end{equation} 
as obtained by transforming $H(t)$ in Eq.\,(\ref{eq:H_x})  according to Eq.\,(\ref{HP}). Likewise, the initial OATS is mapped to a bosonic 
state
\begin{equation}
    |{\rm OATS}\rangle_{\hat{z}} \mapsto |\Psi\rangle= e^{-i\beta (J-a^{\dagger}a)} e^{-i \kappa  \hat{x}^2} |0\rangle, \qquad |{\rm CSS}\rangle_{\hat{z}} \mapsto|0\rangle, \qquad  \kappa= J \mu\label{transfOATS},
\end{equation}
which is Gaussian in a sense we specify below.

\subsubsection{Noise-averaged evolution and covariance matrix.}
The initial state and the noisy dynamics are conveniently described in phase space. We make this mapping explicit via the Weyl transform, which associates to each operator $\mathcal{A}$ a phase-space function $f_{\mathcal A}(x,p)$:
\begin{equation}
\qquad \qquad f_\mathcal{A}(x, p) = \frac{1}{\pi} \int_{-\infty}^{\infty} \!\!\dd y \, \langle x+y | \mathcal{A} | x - y \rangle \, e^{-2 i p y}. 
\label{WeylT}
\end{equation}
This yields an alternative formulation of quantum mechanics that closely parallels classical phase-space dynamics and, in our setting, facilitates intuition.
Applying Eq.\,(\ref{WeylT}) to $\rho_0= |\Psi\rangle \langle \Psi|$ we obtain the Wigner function,  an equivalent representation  of the state in terms of $(x,p)$.  The Wigner function of input $|\Psi\rangle$ is zero-mean and Gaussian:
\begin{equation}
    \qquad \quad W_0(x,p)= \frac{1}{2\pi} \exp \left\{ - \frac{1}{2}  \left[ \delta^{-1} x^2 + \delta(p + 2 \eta x)^2 \right]\right\}. \label{W0}
\end{equation}
The parameters $\delta$ and $\eta$ denote a quadrature squeezing and quadrature shear, respectively, depending on the input variables $(\beta,\mu)$:
\begin{equation}
\delta=1+4\kappa^2 \sin^2(\beta)-2\kappa\sin(2\beta),\quad
\eta=\frac{\kappa^2\sin(2\beta)-\kappa\cos(2\beta)}{\delta}.\label{eq:deltaeta}
\end{equation}
Equivalently, $W_0(x,p)=\mathcal{N}_{\vec{r}}({\vec \mu}_0, \Sigma_0)$, with ${\vec{r}}=(x,p)^{\top}$. The bivariate normal distribution has mean  vector and covariance matrix given by:
\begin{equation}
\begin{array}{rcl}
\vec{\mu}_0
&=&
\left(
\begin{array}{c}
\langle \hat{x} \rangle \\
\langle \hat{p} \rangle
\end{array}
\right)
=
\left(
\begin{array}{c}
0 \\
0
\end{array}
\right),
\\[2ex]
\Sigma_0
&=&
\left(
\begin{array}{cc}
\Delta^2 \hat{x} & {\rm Cov}(x,p) \\
{\rm Cov}(x,p) & \Delta^2 \hat{p}
\end{array}
\right)
=
\left(
\begin{array}{cc}
\delta & -2\eta\delta \\
-2\eta\delta & \delta^{-1}+4\eta^2\delta
\end{array}
\right).
\end{array}
\label{in}
\end{equation}
for $\Delta^2 \hat{y}= \langle \hat{y}^2 \rangle - \langle \hat{y}\rangle^2 $, $y \in \{\hat{x}, \hat{p}\}$, ${\rm Cov} (\hat{x},\hat{p})=\frac{1}{2} \{ \hat{x}, \hat{p}\}-  \langle \hat{x} \rangle \langle \hat{p} \rangle$.

Since the transformed Hamiltonian in Eq.\,(\ref{HHP}) is linear in $\hat x$, the evolution displaces the state along $\hat p$. We can fully describe the state through the drift of the mean and the diffusion of the covariance of a bivariate normal distribution. 
For a single realization of the stochastic process, the mean drifts by a deterministic amount $bt$ and a stochastic contribution $\lambda$,
\begin{equation*}
  \qquad \quad  {\vec\mu}_t(b)= \sqrt{J} \begin{pmatrix}
        0\\
        b t+ \lambda
    \end{pmatrix}, \qquad \lambda(t) = \int_0^t\!\!\dd s \,\xi(s),
\end{equation*}
whereas the covariance matrix remains invariant.  
Fluctuations are, by assumption, also Gaussian, thus the random phase accumulation can be averaged by computing the expectation value of  $W_t(x,p)= \mathcal{N}_{\vec{r}}({\vec\mu}_t, \Sigma_0)$ with respect to the probability 
distribution $\lambda(t) \sim\mathcal N_{\lambda}(0,\chi(t))$, in analogy to Eq.\,(\ref{eq:stochastic_map}) for a discrete qubit setting and finite-dimensional Hilbert space. 
Integration over $\lambda$ resolves the random component of the mean vector. Meanwhile, the covariance matrix picks up an additive, rank-one  contribution that induced diffusion along the $\hat{p}$ axis:
\begin{equation}
\begin{array}{rcl}
\!\!\!\!\!\!\!\!\!\!\!\!\!\tilde{\vec{\mu}}_t(b)
&=&
\sqrt{J}
\left(
\begin{array}{c}
0 \\
b t
\end{array}
\right),
\qquad
\Sigma_t
=
\left(
\begin{array}{cc}
\delta & -2\eta\delta \\
-2\eta\delta & \delta^{-1}+4\eta^2\delta
\end{array}
\right)
+ J\chi(t)
\left(
\begin{array}{cc}
0 & 0 \\
0 & 1
\end{array}
\right).
\end{array}
\label{final}
\end{equation}
The noise-averaged state is $\overline{W}_t(x,p)= \mathcal{N}_{{\bf r}}({\tilde{ \vec \mu}}_t, \Sigma_t)$.

\subsection{Optimal interferometric protocols}

\begin{figure*}[t!]
\centering
\hspace*{18mm}\includegraphics[width=15cm]{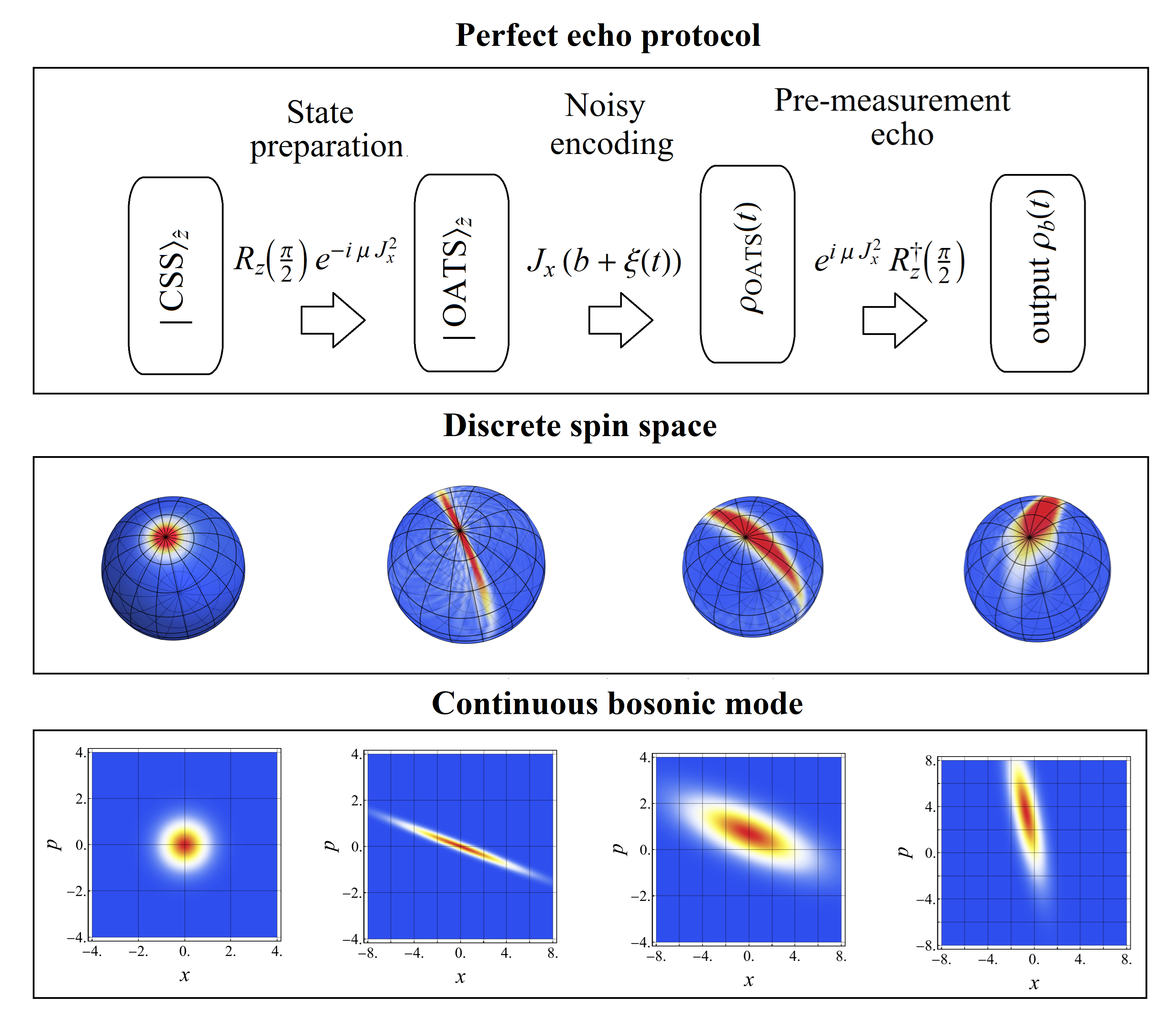}
 \vspace*{-6mm}
\caption{{\bf Optimal interferometric protocol using spin-squeezed states subject to collective dephasing.} 
Top row: Noisy generalized Ramsey protocol using $N=2J$ quantum sensors, with precision saturating the scaling bound in Eq.\,(\ref{eq:QFI_summary_bounds}). The input state is prepared by successively applying a squeezing $e^{-i \mu J_x^2}$ $(\mu=\sqrt{N})$ and rotation $e^{-i \frac{\pi}{2}J_z}$ to a CSS, resulting in the input state $U_{\rm sq} |\rm CSS\rangle_{\hat{z}}=|\rm OATS\rangle_{\hat{z}}$. The signal is imprinted after evolution under the noisy Hamiltonian $H(t)= J_x (b+\xi(t))$ for an encoding time $t$. Fluctuations $\xi(t)$ induce diffusion of strength $J \chi_0^n (\omega_c t)^n$ along the $J_y$ axis. Finally, $U_{\rm sq}^{\dagger}$ un-does the initial squeezing, resulting in output $\rho_b(t)$. Measurement of the collective spin component $J_y$ is optimal.  The second and third rows represent the evolution of the state in phase space  before and after performing the HP transformation, respectively. Parameters:  $J=50$, $b=1$, $n=2$, $\omega_c=1$, $\chi_0^2=\frac{1}{4}$, $t=0.1$.}
\label{QCRB} 
\end{figure*}

For single-mode Gaussian states with mean displacement $\vec{\mu}_t(b)$ and parameter-independent covariance matrix $\Sigma_t$, 
the single-shot SLD and QFI take the form \cite{Pinel,HaidongReview}
\begin{equation}
L= \partial_b\vec{\mu}_t(b)^{\,\top}\,\Sigma_t^{-1} \,{\vec{r}}, \qquad \QFI(t)
=
\partial_b \vec{\mu}_t(b)^{\,\top}\,\Sigma_t^{-1}\,\partial_b \vec{\mu}_t(b).
\end{equation}
Evaluating for the mean vector and covariance matrix in Eq.\,(\ref{final}), the SLD becomes
\begin{equation}
    L= t \frac{\sqrt{J} \delta}{1+ J \delta \,\chi(t)} (\hat{p}+2 \eta \hat{x}) \propto (\hat{p}+2 \eta \hat{x}).
\end{equation}
Therefore, measurement of $O=\hat{p}+2 \eta \hat{x}$ saturates the QCRB, independent of quadrature squeezing. Mapping back to spin space, $O$ corresponds to implementing an anti-squeezing operation of strength $\eta$ before readout of a collective spin component: $O= e^{-i \eta J_x^2} \,J_y \,e^{i \eta J_x^2}$. This is illustrated in Fig.\,\ref{QCRB}, where the anti-squeezing is factored into the output state. This type of generalized Ramsey protocols have been investigated in the absence of noise during encoding \cite{SchulteEchoes,Monika}. By contrast, the total QFI,
\begin{equation}
F_{Q\,\rm tot}^{\rm OATS}(N,T)
=
\frac{J \delta\, T t}{1+J \delta\,\chi(t)} =\frac{N \delta\, T t}{2+N \delta\,\chi(t)},
\label{eq:QFI_tot_OATS}
\end{equation}
is a function of $\delta$ but is completely independent of quadrature shear.
Equation\,(\ref{eq:QFI_tot_OATS}) can be optimized with respect to time in the short-time limit $(\omega_c t \ll 1)$, assuming a power-law behavior for the decay coefficient. Once again, this splits into two cases.


\paragraph{Markovian noise ($n=1$).} When $\chi(t)= \chi_0 (\omega_c t)$, Eq.\,(\ref{eq:QFI_tot_OATS}) is a strictly increasing function of the encoding time. The upper bound
\begin{equation}
    F_{Q\,\rm tot}^{\rm OATS\,}(N,T)= \frac{N \delta T t}{2+N \delta \chi_0 \omega_c t} \leq T \,\sup_{t>0} \frac{N \delta  t}{2+N \delta \chi_0 \omega_c t}= \frac{T}{\chi_0\omega_c},
    \label{QFImark}
\end{equation}
is asymptotically tight and approximately saturated by measuring at times  $(\chi_0 N \delta)^{-1} \lesssim \, \omega_c t$. 
As  Eq.\,(\ref{QFImark}) is state-independent, squeezing does not improve  precision over the coherent spin state.

\paragraph{Colored noise ($n \geq 2$).} Maximizing Eq. (\ref{eq:QFI_tot_OATS}) for $\chi(t) \simeq \chi_0^n (\omega_c t)^n$, we find a global optimal encoding time and  corresponding maximal QFI:
\begin{equation}
\begin{array}{rcl}
t^\star_{\rm OATS}
&=&
\frac{1}{\omega_c \chi_0}
\left[
\frac{2}{(n-1)}
\right]^{1/n}
(N\,\delta)^{-1/n},
\\[2ex]
F_{Q\,\rm tot}^{\rm OATS}(N,T)
&\leq&
\frac{T}{\omega_c \chi_0}
\frac{(n-1)^{\frac{n-1}{n}}}{n}\,
\left(\frac{N}{2}\delta\right)^{-\frac{1-n}{n}}.
\end{array}
\label{eq:QFI_OATS_opt}
\end{equation}
The bound depends on the input only through the quadrature squeezing $\delta$, given in Eq.\,(\ref{eq:deltaeta}).

\smallskip

It is clear from Eq.\,(\ref{eq:QFI_OATS_opt}) that optimal performance requires $\delta$ to be as big as possible while still obeying the low excitation condition $\langle a^{\dagger}a \rangle \ll 2J$. This corresponds to the values $(\mu^{\star}= N^{-1/2}, \beta^{\star}= -\pi/2)$, such that $\delta^{\star}= N$ -- at the boundary of the HP regime of validity. Computing the quadrature shear, we find that $\eta^{\star}=- J \mu$ and, in spin space,  the optimal measurement consists of exactly undoing the initial squeezing: $O=e^{i \mu J_x^2 } J_y e^{-i \mu J_x^2}$. This input state and detection strategy define a \emph{perfect echo} protocol, shown to reach HS in the noiseless setting \cite{Monika}. Collective dephasing does not alter the structure of the optimal interferometric protocol, only the achievable scaling. Table\,\ref{table:0} summarizes the asymptotic performance of three representative members of the OATS class, Eq.\,(\ref{eq:OATS_def}). 

\begin{table*}[t]
 \centering
\begin{tabular}{||c| c c c c c ||} \hline
Initial & Squeezing &  Rotation  &  QCRB   & QCRB & \; QCRB   \\ [0.5ex] 
state & angle, $\mu$  & angle, $\beta$  & colored 
noise  & white noise & noiseless \\ [0.5ex] 
 \hline\hline
$\rm CSS$ & 0 & $0$  &  $\left(6 \,\sqrt{3}\, \right)^{1/2}  N^{-5/4} $ & $1$ & $N^{-1/2}$\\
$\;\;\rm  KU\, OATS\;\;$ & $\;\,3^{1/6}  N^{-2/3}$ & $\;\;\pi/2- 3^{-1/6}\,N^{-1/3}$ & $\sqrt{2}\,(1/3)^{1/12}\, N^{-5/12}  $  & $1$  & $N^{-5/6}$\\
 PE\,OATS & $N^{-1/2}$ & $-\pi/2$  & $2^{1/2}\,N^{-1/2}$ & $1$  & $\sqrt{e}\,N^{-1}$\\ 
 GHZ &  $-$ & $-$ &  $\,(2 e)^{1/4} N^{-1/2}$ & $\sqrt{e}$ & $N^{-1}$ \\ [1ex] 
 \hline
\end{tabular}
\vspace*{1mm}
\caption{
{\bf Asymptotic interferometric performance of input states for linear signal encoding  under collective dephasing.}  The dynamics are generated by $H(t)= (b+ \xi(t)) \hat{x}$. In the first three rows, the input is given by Eq.\,(\ref{eq:OATS_def}), with squeezing and rotation angles $\mu$,  $\beta$ specified in the second and third columns. KU denotes the OATS with initial minimal $\hat{y}$-dispersion considered by Kitagawa and Ueda \cite{Kita1993}, whereas PE denotes OATS used in echo-based protocols \cite{Monika}. The time-optimized QCRB (in units of $(\chi_0 \omega_c/T)^{1/2}$) is provided for colored noise stationary, white noise (fourth and fifth columns) and in the noiseless limit (last column, units of $(T\, t)^{1/2}$). Precision bounds for the $N$-entangled GHZ state (fourth row) are also included.}
\label{table:0}
\end{table*}

\section{Fundamental precision bounds under collective dephasing and control}
\label{openloop}

\subsection{Open-loop control assumptions}

Supplementing the sensors with open-loop control does not allow precision beyond the HS limit in noiseless settings \cite{boixo2007}. However, when the system is subject to dephasing, both pulsed rotations and continuous driving have been shown to improve the {\em single-shot} QFI relative to the uncontrolled evolution for systems comprising few qubits \cite{Haidong2017,zhai2023control2}. An outstanding open question is whether such gains scale with the probe number $N$. In this section, we provide an answer within the setting introduced in Sec.\,\ref{sett}. Concretely, we assume that the available control is collective, hence it leaves the dynamics permutation-invariant and, as before, we work with a fixed total runtime $T$ to carry out the protocol and a fixed (parametric) large number of probes $N$.

We model the external control in terms of sequences of ideal instantaneous pulses, with no restrictions on their number $Q$ or timing separation $0<t_1<\ldots t_Q<t$ within the encoding period. While these assumptions are highly idealized, allowing for this level for control enables the identification of fundamental limits.  
This philosophy is close in spirit to the ``full and fast quantum control'' setting of \cite{Sekatski}, where ancillary resources and QEC are additionally considered. The controlled noisy Hamiltonian is then 
\begin{eqnarray}
H(t) &=& H_0(t) + H_{\rm ctrl}(t), \qquad
H_{\rm ctrl}(t)
= \sum_{k=1}^{Q} \theta_k \, \delta(t - t_k)\,
\mathcal{O}_k.
\label{Ht}
\end{eqnarray}
where $\theta_k$ are arbitrary rotation angles.
By assumption, the Hermitian operators $\mathcal{O}_k$ act on the permutation-invariant (collective) sector, so they can be arbitrary functions of the spin components $(J_x,J_y,J_z)$. Any no-go result derived in this setting applies to more restricted control families, for instance, linear collective control $H_{\rm ctrl}(t)= u_x(t) J_x+u_y(t) J_y+u_z(t) J_z$. In fact, the pulsed model in Eq.~(\ref{Ht}) also captures time-dependent continuous driving as a limiting case. Concretely, for a control Hamiltonian with continuous control inputs, of the form
\begin{equation}
\qquad \qquad H_{\rm ctrl}(t)=\sum_{k=1}^{Q} u_k(t)\,\mathcal O_k, \quad u_k(t) \in {\mathbb R}, \,\forall t, 
\label{eq:continuous_ctrl}
\end{equation}
we can approximate the time-ordered propagator by a product formula on a time-slicing of the finite interval $[0,t]$ with step $\Delta t$. We can interpret each step as an instantaneous pulse of area $\Delta t\,u_k(t_\ell)$, applied at time $t_\ell$, which yields Eq.\,(\ref{Ht}) in the infinitesimal limit $\Delta t\to 0$ (see Appendix~\ref{app:continuous_limit} for more detail). 

\subsection{Random-unitary representation of controlled dynamics}
\label{subsec:multilabel}

Let $0=t_0<t_1<\cdots<t_{Q'}=t$ be the segmentation induced by the pulses, with $Q'\equiv Q+1$ free-evolution intervals and $\Delta t_j\equiv t_{j+1}-t_j$ a free-evolution segment, defining the vector $\Delta \vec t \in \mathbb{R}^{Q'}$. Similarly, let us collect the random phase accumulation in each segment into a vector $\vec\lambda\in\mathbb{R}^{Q'}$, with components  $\lambda_j \equiv \int_{t_j}^{t_{j+1}}\!\dd s\,\xi(s)$. Since the process $\xi(t)$ is Gaussian, $\vec\lambda$ is a centered multivariate Gaussian random variable, $\vec\lambda\sim\mathcal N_{\vec\lambda}(\vec 0,\Sigma(t))$. Its covariance matrix $\Sigma \succ 0$ has entries given by $\chi_{ij}= \langle \lambda_i \lambda_j \rangle$. To characterize how control acts on these labels, we introduce the cumulative pulses and the corresponding toggling-frame generators
\begin{equation}
\quad \qquad P_0\equiv\mathbb I,
\qquad
P_j\equiv \mathcal U^{(j-1)}\cdots \mathcal U^{(0)}, \qquad G_j\equiv P_jJ_zP_j^\dagger.
\label{eq:toggling_generators}
\end{equation}
with $\mathcal{U}^{(k)} \equiv e^{-i\,\theta_k \,\mathcal{O}_k}$ being the $k$th instantaneous pulse. A consecutive block of segments $\mathcal{B}=\{r,\ldots,s\}$ of length $\Delta t_\mathcal B \equiv t_{s+1}-t_r$ is \emph{dephasing-preserving} (DP) if there exist a Hermitian collective operator $G_{\mathcal B}$ and signs $y_j\in\{\pm1\}$ such that $G_j=y_j\,G_{\mathcal B},$ $j\in\mathcal B.$ In such a case, all the generators commute and the block depends on the segment labels only through the combinations
\begin{equation}
\quad \qquad\lambda_{\mathrm{eff}}^{(\mathcal B)}
\equiv \sum_{j=r}^{s} y_j\lambda_j,
\qquad
\Delta t_{\mathrm{eff}}^{(\mathcal B)}
\equiv \sum_{j=r}^{s} y_j\Delta t_j,
\label{eq:effective_label_block}
\end{equation}
with the resulting variable $\lambda_{\mathrm{eff}}^{(\mathcal B)}$ being also Gaussian.

A general control sequence may contain DP blocks -- as arising, for instance, from sequences of $\pi_x$ or $\pi_y$ pulses most commonly employed in single-axis DD protocols \cite{violalloyd, ViolaDD, MultiDD} -- interspersed with genuinely non-commuting segments -- whereby ${\cal U}^{(k)}$ involves rotations with arbitrary axis and angle, or possibly nonlinear control operations. Compressing each maximal DP block to its effective label and keeping the remaining segments unchanged yields an ordered reduced label vector $\mathcal{S}\vec{\lambda}\equiv \vec\Lambda=(\Lambda_1,\dots,\Lambda_{\widetilde Q})^\top $. Here, $\mathcal{S}\in\mathbb{R}^{\widetilde Q\times Q'}$ is a full-row-rank matrix defining the linear compression, such that $\Lambda_{\alpha}=\lambda_k$ for a single segment, and $\Lambda_\alpha= \lambda_{\mathrm{eff}}^{(\mathcal B)}$ for a maximal DP block. The compressed random vector remains  Gaussian-distributed, $\vec\Lambda\sim\mathcal N(\vec 0,\widetilde\Sigma(t))$, with a covariance matrix $\widetilde \Sigma(t)\equiv  \mathcal{S} \,\Sigma(t) \,\,\mathcal{S}^{\top}\succ 0 \in \mathbb{R}^{\widetilde Q \times \widetilde Q}$, with entries $\chi_{ij}(t)  =\langle \Lambda_i \Lambda_j \rangle.$ The noise-averaged evolution in the presence of the applied control can be written as a multilabel convex unitary mixture:
\begin{equation}
\rho_b(t)=\int_{\mathbb R^{\widetilde Q}}\dd\vec\Lambda\;
\mathcal N(\vec 0,\widetilde\Sigma(t))\;
U_{\vec\Lambda}(t)\,\rho_0\,U^{\dagger}_{\vec\Lambda}(t), \qquad U_{\vec\Lambda}(t)
\equiv \prod_{\alpha=1}^{\widetilde Q}\mathcal V_\alpha(\Lambda_\alpha),
\label{eq:compressed_multilabelRU}
\end{equation}
where $\mathcal V_\alpha$ is the rank-one factor
$e^{-iG_{\mathcal B}(b\,\Delta t_{\mathrm{eff}}^{(\mathcal B)}+\Lambda_\alpha)}$ for a DP block and the original single-segment factor $\mathcal U^{(k)}e^{-iJ_z(b\,\Delta t_k+\Lambda_\alpha)}$ otherwise.

\subsection{No-go for superclassical precision scaling under controlled collective dephasing}
\label{subsec:DDnogo}

Our goal is now to derive a {\em state-independent} upper bound on the time-optimized QFI in the controlled setting defined by Eq.\,(\ref{Ht}). We follow and generalize the fidelity-based strategy of Ref.~\cite{DurDD}, consisting of five key steps: 
(i) represent the controlled evolution as a Gaussian mixture over the segment-integrated labels $\vec\Lambda$; 
(ii) reparameterize the labels so that the two compared states share the same unitary branches; 
(iii) lower bound the quantum fidelity by the classical fidelity between the corresponding Gaussian weights; 
(iv) convert this lower bound into an upper bound on the QFI; 
(v) (for stationary colored noise) combine with the noiseless HS ceiling and optimize over the interrogation time $t$.

\subsubsection{QFI bounds under control.}
Following the above outline, we prove two no-go results precluding sub-SQL performance when the noise is stationary and i) temporally uncorrelated with a white spectrum or ii) temporally correlated with a colored spectrum. In the latter case, relevant constant-factor improvements in finite-$N$ experiments are not discarded. Both results hinge on the following key result: 

\begin{theorem}
[Upper bound to the QFI] 
\label{thmquadf} 
Let $0< t_1<t_2 <\ldots <t_{ Q}<t$ be the segmentation of an arbitrary sequence of instantaneous collective pulses whose output state is given by Eq.\,(\ref{eq:compressed_multilabelRU}). Then, the single-shot QFI is upper-bounded by a quadratic form as follows:
$$F_Q[\rho_b(t)] \leq \Delta {\vec{t}}^{{}\,\top} \Sigma^{-1} \Delta \vec t,$$ where $\Sigma \succ 0 \in \mathbb{R}^{Q' \times Q'}$ and $\Delta \vec t \in \mathbb {R}^{Q'}$  are the \emph{uncompressed} covariance matrix of the random phase accumulation vector $\vec \lambda$ and the segmentation vector, respectively. 
\end{theorem}

\begin{proof}
As in Sec.\,\ref{QFIg}, we express the QFI in terms of the fidelity of two outputs with infinitesimally differing values of the target parameter, Eq.\,(\ref{QFIFid3}):
\begin{equation}
\qquad \mathcal{F}\bigl(\rho_b(t),\rho_{b+\delta b}(t)\bigr)
=1-\frac{1}{8}F_Q[\rho_b(t)]\,\delta b^2 + \mathcal{O}(\delta b^3).
\label{eq:fidelityToQFI}
\end{equation}
Let us introduce $\Delta \vec{\tau}\equiv \mathcal{S} \Delta \vec t$, the vector with signed components $ \Delta \tau_{\alpha}=\Delta t_k$ for a single segment and $\Delta \tau_{\alpha}=\Delta  t_{\rm eff}^{(\mathcal{B})}< \Delta t_{\mathcal{B}}$ for a maximal DP block. We can then write the states $\rho_b(t)$ and $\rho_{b+\delta b}(t)$ in the multilabel unitary mixture compressed representation, Eq. \,(\ref{eq:compressed_multilabelRU}), and perform the following change of variables in $\rho_{b+\delta b}(t)$: 
\begin{equation}
\qquad \quad \Lambda_{\alpha} \mapsto \Lambda_{\alpha} + \delta b\,\Delta \tau_{\alpha}, \qquad \alpha=1,\ldots, \widetilde Q.
\label{eq:lambdashift}
\end{equation}
This aligns the ``unitary branches'' between the two states:  both become mixtures over the \emph{same} family $\{U_{\vec\Lambda}(t)\rho_0 U_{\vec\Lambda}^{\dagger}(t)\}$, at the cost of the Gaussian weight in  $\rho_{b+\delta b}(t)$ acquiring a mean shift
$\dd \widetilde \Phi(\delta b) \equiv \delta b\,(\Delta \tau_1,\ldots,\Delta \tau_{\widetilde Q})^{\top}$. Compactly, we may write 
\begin{equation}
\rho_b(t) = \int d\vec\Lambda\;\mathcal{N}_{\vec\Lambda}(\vec 0,\widetilde \Sigma)\;\rho_{\vec\Lambda},\quad 
\rho_{b+\delta b}(t) = \int d\vec\Lambda\;\mathcal{N}_{\vec\Lambda}(\dd \widetilde\Phi(\delta b),\widetilde \Sigma)\;\rho_{\vec\Lambda},
\label{eq:alignedMixtures}
\end{equation}
where $\rho_{\vec{\Lambda}}= U_{\vec\Lambda}(t)\,\rho_0\,U^{\dagger}_{\vec\Lambda}(t)$ is an element of the mixture.
We may use the strong concavity property of the 
fidelity, 
\begin{equation}
\mathcal{F}\bigg(\sum_k p_k \rho_k, \sum_j q_j \sigma_j\bigg) \geq \sum_k \sqrt{p_k q_k} \,\mathcal{F}(\rho_k, \sigma_k),    
\end{equation}
 to upper bound $\mathcal{F}\!\left(\rho_{b},\rho_{b+\delta b}\right)$. For two Gaussians with equal covariance $\widetilde \Sigma$, the overlap is
\begin{equation}
\!\!\!\!\!\!\!\!\!\!\!\!\!\!\!\!\!\!\!\!\!\!\!\!\!\!\!\!\!\!\!\!\!\mathcal{F}\!\left(\rho_{b},\rho_{b+\delta b}\right)\geq \int d\vec\Lambda\;
\sqrt{\mathcal{N}_{\vec\Lambda}(\vec 0,\widetilde \Sigma)\,\mathcal{N}_{\vec\Lambda}(\dd\Phi,\widetilde \Sigma)} \,\underbrace{\mathcal{F}(\rho_{\vec{\Lambda}},\rho_{\vec{\Lambda}})}_{=1}
=
\exp\!\left(-\frac{1}{8}\,\dd \Phi^{\top} \widetilde\Sigma^{-1}\dd\Phi\right).
\label{eq:gaussianOverlapClosedForm}
\end{equation}
Expanding at small $\delta b$ gives the desired lower bound in terms of a quadratic form:
\begin{equation}
\mathcal{F}\bigl(\rho_b(t),\rho_{b+\delta b}(t)\bigr)
\ge 1-\frac{1}{8}\,(\Delta\vec \tau^{\,\top}\widetilde \Sigma^{-1}\Delta\vec \tau)\,\delta b^2 + O(\delta b^3).
\label{eq:fidelityExpansion}
\end{equation}

We now leverage a result whose proof we defer to Appendix \ref{lemmaproof}:

\begin{lemma} [Monotonicity of the Gaussian quadratic form under linear compression]
\label{lem:compression_monotonicity}

Let $\vec\lambda \in \mathbb{R}^{Q'}$ be a zero-mean Gaussian random vector with covariance
$\Sigma \succ0\in \mathbb{R}^{Q'\times Q'}$. If $\mathcal{S}\in\mathbb{R}^{\widetilde Q\times Q'}$ is a full-row-rank matrix that defines a linear compression $\vec\Lambda = \mathcal{S}\,\vec\lambda$, with covariance $\widetilde \Sigma = \mathcal{S} \,\Sigma \,\mathcal{S}^\top$, the quadratic form appearing in the Gaussian overlap satisfies
\begin{equation}
(\mathcal{S}\,\Delta \vec t)^\top \widetilde \Sigma^{-1} (\mathcal{S}\,\Delta \vec t)= \,\Delta \vec \tau^\top \widetilde \Sigma^{-1} \,\Delta \vec \tau \;\le\;
\Delta \vec t^\top \Sigma^{-1} \Delta \vec t.
\label{eq:compression_monotonicity}
\end{equation}
\end{lemma}

\noindent 
It follows that any sequence containing DP blocks is upper-bounded by a sequence with the same segmentation that does not contain such blocks. We can then assume, without loss of generality, that {\em no DP block} is present in what follows. 

Replacing Eq.\,(\ref{eq:fidelityExpansion}) in Eq.\,(\ref{eq:fidelityToQFI}) then leads to an upper bound on the single-shot QFI, which we further bound using Lemma \ref{lem:quadform}: 
\begin{equation}
\!\!\!\!\!\!\!\!\!\!\!\!\!\!\!\!\!\!\!\!\!\!F_Q[\rho_b(t)] \le \Delta\vec \tau^{\top} \widetilde\Sigma^{-1}\Delta\vec \tau \leq \Delta \vec{t}^{\top} \Sigma^{-1} \Delta \vec{t} \quad \Rightarrow \quad F_Q^{\rm tot}[\rho_b(t)]\le
\frac{T}{t}\,\Delta\vec t^{\top }\Sigma^{-1}\Delta  \vec t,
\label{eq:minUpperBounds}
\end{equation}
where the bound for the total QFI is obtained by considering $\nu=T/t$ repetitions. 
\end{proof}

A relevant particular case arises when the control is DP throughout the whole encoding interval -- for instance, when a train of instantaneous $\pi_x$ or $\pi_y$ pulses is applied, as mentioned earlier. In the limit of arbitrarily fast control, the resulting DD schemes are capable in principle of suppressing decoherence to arbitrary accuracy \cite{violalloyd,ViolaDD} (though, importantly, a finite accuracy arises if a finite minimum-separation constraint between pulses is imposed \cite{ViolaLimits}). A key observation is that any DP sequence that improves the short-time (quadratic) behavior of the decay coefficient  $\chi_{\rm DD}(t)= \langle \lambda_{\rm eff}^{(\mathcal{B})} \lambda_{\rm eff}^{(\mathcal{B})}\rangle$ 
must necessarily {\em cancel the signal as well}. This is easy to verify by writing  $\chi_{\rm DD}(t)$ in the frequency domain:
\begin{equation*}
    \chi_{\rm DD}(t) =\int_{- \infty}^{\infty} \dd \omega\, \frac{S(\omega)}{2 \pi}\, 
\tilde{F}(t, \omega), \quad 
\tilde{F}(t,\omega)= \left| \int_{0}^t \!\dd s\, y(s)\, e^{-i \omega s} \right|^2.
\end{equation*}
Expanding $\tilde{F}(t,\omega)$ at short times and imposing a vanishing filtering-order at $\omega=0$ 
requires $\int_0^t \!\dd s\, y(s)=0$ \cite{Paz2014}, which in turn leads to $b \Delta t_{\rm eff}^{\mathcal{B}}=0$, rendering the procedure useless for estimation. On the other hand, DP sequences with a nonzero filtering order  at zero partially retaining the signal but are unable to alter the short-time behavior of $\chi_{\rm DD}(t)$. The resulting output state $\rho_b(t)$ can be written as a convex unitary mixture depending on a {\em single} effective label $\lambda_{\rm eff}^{(\mathcal{B})}$, given by Eq.\,(\ref{eq:effective_label_block}), and the upper bound in Eq.\,(\ref{eq:minUpperBounds}) again directly applies\footnote{We note that, in a DP control scenario, it is also possible to derive precision bounds precluding improvements over free evolution by using the purification approach pursued in Sec.\,\ref{general} with appropriate modification, see \cite{thesis}.}.

\subsubsection{No-go theorems.}
Theorem \ref{thmquadf} may be used to establish the no-go results we anticipated. 

\begin{corollary}
\label{thm:white_nogo}
For frequency estimation subject to stationary, collective, white-noise dephasing, arbitrary collective open-loop control cannot improve the estimation precision beyond the $N$-independent bound derived in the absence of control, Eq.\,(\ref{eq:QFI_summary_bounds}).
\end{corollary}

\begin{proof}

For white noise, the segments are uncorrelated and the corresponding variances add linearly in time, namely, $\chi_{ij}= \chi_0\omega_c\,\Delta t_i\, \delta_{ij} $. Therefore, 
\begin{equation}
\Delta\vec t^{\,\top}\Sigma^{-1}\Delta\vec t
\le
\frac{1}{\chi_0\omega_c}\sum_{i=1}^{Q'} \Delta t_i= \frac{t}{\chi_0\omega_c}.
\label{eq:compressed_white_quadratic_bound}
\end{equation}
Substituting into Eq.~(\ref{eq:minUpperBounds}) yields the state-independent bound
\begin{equation}
F_Q^{\mathrm{tot}}[\rho_b(t)]
\le
\frac{T}{\chi_0\omega_c}
\quad \Rightarrow \quad
\Delta \hat b_{\mathrm{QCRB}}^2
\ge
\frac{\chi_0\omega_c}{T},
\label{eq:whiteFinalMin}
\end{equation}
which coincides with the bound derived in the absence of control in Eq.\,(\ref{eq:QFI_summary_bounds}).
\end{proof}

\begin{corollary}
\label{thm:finite_nogo}
For frequency estimation subject to stationary, collective, temporally correlated dephasing, arbitrary collective open-loop control cannot improve the estimation precision beyond the SQL asymptotic scaling obtained without control.
\end{corollary}

\begin{proof}
The overlap-based argument used in Theorem \ref{thm:white_nogo} yields the general bound in Eq.\,(\ref{eq:minUpperBounds}). We tighten it by combining it with the noiseless HS limit, that is, 
\begin{equation}
F_Q^{\mathrm{tot}}[\rho_b(t)]
\le
\min\!\left(
\frac{T}{t}\,
\Delta\vec t^{\,\top}\Sigma^{-1}\Delta\vec t,\;
T t N^2
\right).
\label{eq:compressed_min_colored}
\end{equation}
To compute the upper bound in Eq.\,(\ref{eq:compressed_min_colored}), we must determine the behavior of the quadratic form $\Delta\vec t^{\,\top}\,  \Sigma^{-1}\Delta \vec t$ in the short-time limit for arbitrary pulse placements. Again, we use a technical result we prove in Appendix \ref{lema1proof}:

\begin{lemma}[Quadratic form bound in the short-time regime]
\label{lem:quadform}
For stationary colored noise, there exists a dimensionless constant $K_{Q'}$, depending only on the number of control segments $Q'$, such that, for $\omega_c t\ll 1$,
\begin{equation}
\Delta\vec t^{\,\top}\widetilde \Sigma^{-1}\Delta\vec t
\;\le\;
K_{Q'}\,\omega_c^{-2}.
\label{eq:quadformBoundShortTime}
\end{equation}
\end{lemma}

\noindent
At short times, $\Sigma^{-1}$ develops a structured, nearly singular form determined by the noise correlation function. The 
idea is to perturbatively expand the inverse covariance in terms of the smallness parameter $\omega_c t$ to approximate the quadratic form at short times.  Substituting Eq.\,(\ref{eq:quadformBoundShortTime}) into Eq.\,(\ref{eq:compressed_min_colored}) yields
\begin{equation}
F_Q^{\mathrm{tot}}(T,t)
\le
\max_{t>0}
\min\!\left(
T t N^2,\;
\frac{T}{t}\,K_{Q'}\,\omega_c^{-2}
\right).
\label{eq:compressed_finiteMinBound}
\end{equation}
The maximizing interrogation time is $t^\star
= K_{Q'}^{1/2}\,\omega_c^{-1}N^{-1},$
which lies in the short-time regime for $N\gg1$. Evaluating the QFI at $t=t^\star$ gives
\begin{equation}
F_Q^{\mathrm{tot}}[\rho_b(t)]
\le
\frac{T}{\omega_c}\,K_{Q'}^{1/2}N
\quad
 \Rightarrow \quad \Delta \hat b_{\mathrm{QCRB}}^2
\ge
\frac{\omega_c}{T}\, K_{Q'}^{-1/2}N^{-1}.
\label{eq:compressed_QFIperf}
\end{equation}
Thus, no scaling gain over the SQL is possible, as claimed. 
\end{proof}

Physically, the short-time behavior of the quadratic form is controlled by the segmentation geometry through the covariance matrix $\Sigma$, and bounded by Eq.\,(\ref{eq:compressed_finiteMinBound}). Open-loop control therefore cannot improve the asymptotic scaling with $N$. The constant entering Eq.\,(\ref{eq:compressed_QFIperf}) encodes a geometric factor arising from the
inversion of the reduced covariance matrix. Its value depends on the spectrum details and typically
grows with the number of pulses, rendering the bound progressively larger
as more pulses are
applied. For a Gaussian-decaying spectrum,
\begin{equation}
S(\omega)=\alpha\,\omega^s\omega_c^{1-s}e^{-(\omega/\omega_c)^2},
\label{Gausspec}
\end{equation}
we are able to compute this prefactor exactly in Appendix~\ref{app:KQ_gaussian_corrected}, finding the result $K_{Q'}\propto Q^{(s+1)/2}$, as illustrated in Fig.~\ref{KQplot}. While no scaling advantage is permitted, this leaves open the possibility for a potentially meaningful improvement in the achievable precision for given $N$ 
when $Q \gg 1$. However, we leave the problem of finding 
a concrete protocol that may realize this advantage to future investigation.

\begin{figure*}[t!]
\centering
\includegraphics[width=16cm]{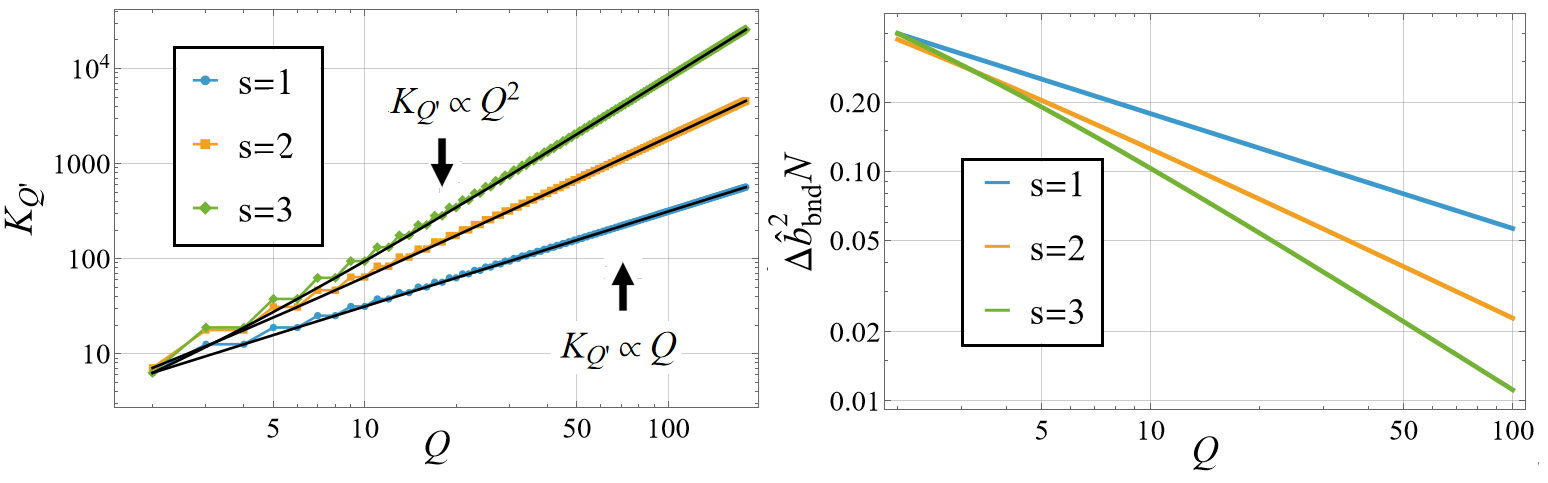}
\vspace*{-3mm}
\caption{{\bf Constant prefactor in the QFI upper-bound as a function of pulse number.}
Left: Behavior of the constant $K_{Q'}$ entering $F_Q^{\rm tot}[\rho_b(t)]$ in Eq.\,(\ref{eq:compressed_QFIperf}) as a function of pulse number $Q$ for a colored spectrum with a Gaussian profile, Eq.\,(\ref{Gausspec}). The prefactor grows as $Q^{(s+1)/2}$.  The colored curves join the exact $K_{Q'}$ values, represented by dots, whereas the black solid lines describe the coarse asymptotic behavior.
Right: Estimation lower bound as a function of pulse number.
Sub-SQL performance is still precluded, but the prefactor $\Delta \hat b_{\rm bnd}^2 N \sim K_{Q'}^{-1}$ in Eq.~(\ref{eq:compressed_QFIperf}) may improve significantly. Parameters: $\alpha=1$, $\omega_c=1$, $T=1$.}
\label{KQplot}
\end{figure*}

\section{Conclusion}
\label{conclusions} 

In this work, we analyzed entanglement-assisted frequency estimation in the presence of stochastic phase fluctuations. These were modeled as zero-mean, Gaussian noise acting {\em collectively} on the probes, and characterized by arbitrary temporal correlations. We derived state-independent lower bounds on the achievable precision when the experimental runtime $T$ is the key resource constraint. The asymptotic scaling of these bounds is entirely determined by the short-time behavior of the decoherence function,
$\chi(t)\simeq \chi_0^n(\omega_c t)^n$, with $n\in\mathbb{N}$.

For linear and quadratic short-time decay, $n\in\{1,2\}$, corresponding respectively to stationary Markovian and colored dephasing noise, the bounds impose a no-go to superclassical metrological performance. By contrast, a quantum advantage remains possible in principle only for colored non-stationary noise with $n\ge 3$. In all cases, the optimal noiseless sequence using spin-squeezed states \cite{Monika}  saturates the corresponding noisy precision bounds.

We further assessed the role of open-loop control based on arbitrary sequences of collective control pulses, including both DD sequences and continuous open-loop control as special instances. 
The resulting dynamics may be exactly represented as a multilabel mixture of unitary channels. For stationary Markovian and colored noise, we proved that allowing for any number of pulses and arbitrary pulse directions cannot improve the asymptotic precision scaling. Instead, performance is ultimately limited by the same bounds that apply in the absence of control, up to a constant pre-factor in the case of temporally correlated noise. Thus, open-loop pulsed control is ineffective as a metrological resource for overcoming the SQL in these regimes.

Several extensions of the present analysis merit further investigation. We expect a similar absence of scaling improvements for non-stationary, temporally correlated noise ($n \ge 3$) under the application of control. Extending the proof, however, requires additional control of high-order spectral moments, which may be unbounded in this regime. It is also natural to conjecture that SQL-type limitations persist for collective control when the system interacts with a bona fide nonclassical bath \cite{FelixPRA,FUR}. In this case, the decay coefficient shares the same functional form as those generated by classical Gaussian fluctuations. Additionally, however, the non-commuting bath operators give rise to noise-induced \emph{phase} contributions. While we expect that these additional  contributions may only 
further degrade the precision and hinder a metrological advantage, a formal proof is lacking as yet. 
Finally, as mentioned, whether a suitably designed control protocol 
is capable of saturating the prefactor advantage allowed by $K_{Q'}$ is an open question worthy of investigation. As the bound to precision improves with the number of pulses, this could result in a meaningful advantage in practical scenarios. 

Altogether, our results identify collective dephasing as a noise model in which optimal interferometric strategies are remarkably rigid. The best noiseless protocols remain optimal, and neither changing the probe initializations nor applying symmetry-preserving, collective open-loop control can fundamentally alter the attainable precision scaling.

\section*{Acknowledgments}

It is a pleasure to thank Leigh Norris, Maryam Mudassar and Natthaphong Wonglakhon for insightful discussions and early contributions related to various aspects of this work. The authors are also indebted to Milad Marvian for valuable input. Work at Dartmouth was supported by the US NSF through Grant No.\,PHY-2013974. Work at Griffith was supported  by the Australian Research Council Discovery Projects Grant No.\,DP210102291. F.R. also gratefully acknowledges support from the DOE EXPRESS award No.\,DE-SC0024685.

\section*{Appendix}

\renewcommand{\theequation}{A.\arabic{equation}}
\setcounter{equation}{0}
\renewcommand{\thesection}{\Alph{section}}
\setcounter{section}{0}

\section{Additional discussion on noise correlation and spectral properties}
\label{app:formal_spectrum}

\subsection{Sufficient condition for short-time quadratic scaling of $\chi(t)$}

Assume that the stochastic process $\xi(t)$ is wide-sense stationary. Then the decay function $\chi(t)$ is quadratic at short time, 
$\chi(t)=C(0)\,t^2+o(t^2),$ $t\rightarrow 0$, if (i) the variance is finite, $C(0)=\langle \xi(s)\xi(s)\rangle < \infty$, and (ii) $C(\tau)$ is continuous at $\tau=0$. 

\paragraph{Proof.}
Assume that assumptions (i) and (ii) hold. For small $\tau$, we can then write
$$ 
C(\tau)=C(0)+\epsilon(\tau),
\qquad 
\epsilon(\tau)\rightarrow[\tau\to0]{}0 .$$
Substituting into Eq.\,(\ref{chit}) in the main text and integrating with respect to $\tau= s-s'$ gives
\begin{equation*}
\chi(t)
=2\int_0^t\!\dd\tau  (t-\tau)\big[C(0)+\epsilon(\tau)\big].
\end{equation*}
The leading contribution is
$$
2C(0)\int_0^t\!\dd\tau \,(t-\tau)
=2C(0)\left[\frac{t^2}{2}\right]
=C(0)t^2 .
$$
The correction satisfies
$$
\left|2\int_0^t\!\dd\tau\, (t-\tau)\epsilon(\tau)\right|
\le 
2t\int_0^t \!\dd\tau\, |\epsilon(\tau)|,
$$
which is $o(t^2)$ because $\epsilon(\tau)\to0$.  Hence the leading behavior is quadratic. 
It follows that finite variance together with continuity of $C(\tau)$ at the origin are sufficient conditions for quadratic short--time scaling.
\hfill \qed

\subsection{On the decay coefficient of an integrated stationary-noise process}

Suppose that the fluctuations of interest arise from the integrated noise of a stationary zero-mean process $\eta(t)$ with spectrum $S_\eta(\omega)$
$$
\xi(t)=\int_0^t\!\!\dd u\, \eta(u) , \qquad S_\eta(\omega)= \int_{-\infty}^{\infty}\!\!\dd \omega \, \langle \eta (\tau) \eta (0)\rangle \,e^{-i \omega \tau}. $$
The decay coefficient is the variance of random phase accumulation $\lambda$:
\begin{equation}
\lambda=\int_{0}^t\!\!\dd s\, \xi(s) \quad \Rightarrow \quad \chi(t)= \langle\lambda^2\rangle= \int_0^t\!\!\dd s\int_0^t\!\!\dd s'\left[\int_0^t\!\!\dd u\int_0^t\!\!\dd u' \langle \eta(u)\eta(u')\rangle\right].
 \label{lX}
\end{equation}
As $\eta(t)$ is stationary by assumption, it has a well-defined spectrum, from which the Wiener-Khinchin relationship emerges:
\begin{equation*}
    \langle \eta(u) \eta(u')\rangle= \frac{1}{2\pi} \int_{-\infty}^{\infty}\!\!\dd \omega\, e^{i \omega \tau} S_\eta(\omega).
\end{equation*}
 Replacing in Eq.\,(\ref{lX}) leads to an expression for the decay coefficient in the frequency domain:
\begin{equation*}
    \chi(t)= \frac{1}{2 \pi} \int_{-\infty}^{\infty}\!\!\dd \omega\, S_\eta(\omega) \, \left| \int_{0}^t\!\!\dd s\int_0^s\!\!\dd u \,e^{i \omega u} \right|^2=\frac{1}{2 \pi} \int_{-\infty}^{\infty}\!\!\dd \omega\, S_\eta(\omega) \, 
    G(t,\omega) 
\end{equation*}
We may write $\xi(t)$ as the increments of a convolution process:
\begin{equation*}
    \xi(t)= \xi_{-\infty}(t)-\xi_{-\infty}(0), \qquad \xi_{-\infty}(t)= \int_{-\infty}^t\!\!\dd s \,\eta(s)= \eta(t) \,*\,\theta(t),
\end{equation*}
where $\theta(t)$ is the Heaviside step function. The Fourier transform of $\xi_{-\infty}(t)$ is the product of the transforms of the functions making the convolution: 
\begin{eqnarray}
  \hat{\xi}_{-\infty}(\omega)= \hat{\eta}(\omega)  \,\hat{\theta}(\omega),\qquad \hat{\theta}(\omega)=  {\rm P.V.} \left( \frac{1}{i \omega}\right) + \pi \delta(\omega),  
  \label{dec1}
\end{eqnarray}
with ${\rm P.V.}$ denoting Cauchy's principal value.
Formally, we can assign $\xi_{-\infty}(t)$ a generalized spectrum from the identity:
\begin{eqnarray*}
    S_{\xi_{-\infty}}(\omega)&= \int_{-\infty}^{\infty}\!\!\dd \omega\, e^{-i\omega \tau } \,\left[\int_{-\infty}^{\infty}\!\!\dd s \int_{-\infty}^{\infty}\!\!\dd s'\,\langle \eta(t+\tau-s) \eta (t-s')\rangle \theta(s) \theta(s') \right] \nonumber\\
    &= \left| \int_{-\infty}^{\infty}\!\!\dd s \,e^{i \omega s} \theta(s) \right|^2 \int_{-\infty}^{\infty}\!\!\dd \omega\, e^{-i\omega (\tau-s+s')} \langle \eta(\tau-s+s')\eta (0)\rangle \nonumber\\
    &= |\hat{\theta}(\omega)|^2 S_\eta (\omega).
\end{eqnarray*}
We emphasize that this object does not define a spectrum in the Wiener-Khinchin sense, as $\xi(t)$ is non-stationary; however, it is a useful representation for computing second moments. Writing the phase accumulation $\lambda$ in terms of the Fourier components of $\xi_{-\infty}(t)$,
\begin{eqnarray*}
    \lambda&= \frac{1}{2 \pi} \left[ \int_{-\infty}^{t}\!\!\dd s  \int_{-\infty}^{\infty}\!\!\dd \omega \,\hat{\xi}_{-\infty}(\omega) e^{-i \omega s}-\int_{-\infty}^{0}\!\!\dd s'  \int_{-\infty}^{\infty}\!\!\dd \omega \,\hat{\xi}_{-\infty}(\omega) e^{-i \omega s'} \right]\nonumber\\
    &= \frac{1}{2 \pi} \int_{-\infty}^{\infty}\!\!\dd \omega\,\hat{\xi}_{-\infty}(\omega)  \left[\frac{1-e^{-i \omega t}-i \omega t}{i \omega}\right]\\&=\frac{1}{2 \pi} \int_{-\infty}^{\infty}\!\!\dd \omega\,\hat{\xi}_{-\infty}(\omega)  B(t,\omega),
\end{eqnarray*}
and replacing $\hat{\xi}_{-\infty}(\omega)$ from Eq.\,(\ref{dec1}), we see that the $\delta(\omega)$ contribution vanishes, as $B(t,\omega)$ is zero for $t=0$. We may then express $\chi(t)=\langle \lambda^2 \rangle$ in terms of the ${\rm P.V} (1/i \omega)$ contribution. Hence, we formally identify
\begin{equation*}
    S_{\xi}(\omega)\equiv \left|{\rm P.V.} \left(\frac{1}{i \omega}\right) \right|^2 \! \! S_\eta(\omega)=\frac{S_\eta (\omega)}{\omega^2} .
\end{equation*}

\section{Computation of energy moments 
for precision bound based on variational approach}
\label{varcalc}

We evaluate  the moments $\langle \mathit{H} \rangle $ and $ \langle \mathit{H}^2 \rangle $, where
\begin{equation*}
    H= H_{\rm SB}- h_{\rm B}= t J_z -\zeta\,t\,\frac{(a^\dagger+a)}{2\sqrt{\chi(t)}},
\end{equation*}
and the expectation is taken with respect to system-bath state
\begin{equation*}
    |\vartheta\rangle \otimes |0\rangle_{\rm B} =  \sum_{m=-J}^J c_m |J, m \rangle \otimes |0\rangle_{\rm B}.
\end{equation*}
 For $\langle H \rangle$, expectation $\langle H_{\rm SB}\rangle=t \,\langle J_z \rangle$ is straightforward to compute, whereas $\langle h_{\rm B}  \rangle$ gives:
\begin{eqnarray*}
\langle h_{\rm B} \rangle &= \frac{\zeta\, t }{\,2  \sqrt{ \chi(t)}} \sum_m |c_{m m}|^2\, _{\rm B} \big \langle \sqrt{\, \chi(t)}\, m|   (a^{\dagger} + a) | \sqrt{\, \chi(t)}\, m \big\rangle_{\rm B} \nonumber \\
&= \frac{\zeta\, t}{2 \sqrt{ \,\chi(t)}} \bigg(\sum_m |c_{mm}|^2 \, 2  m\bigg) 2\,\sqrt{\, \chi(t)} =\zeta\, t \, \langle J_z \rangle. \nonumber
\end{eqnarray*}
In the first line we have taken the partial trace with respect to system $\rm S$, which leaves us with the coherent states $| \sqrt{\chi(t)}\, m \rangle_{\rm B} =e^{\sqrt{\,\chi(t)} (a^{\dagger} -  a) m} |0\rangle_{\rm B} $. In the second line, we have used the coherent states property $a |\alpha\rangle = \alpha |\alpha\rangle$. Overall, we have 
$\langle \mathit{H} \rangle= t (1-\zeta) \langle J_z \rangle  .$

The second moment  $\langle \mathit{H}^2 \rangle = \langle H_{\rm SB}^2\rangle - \langle H_{\rm SB}\; h_{\rm B} \rangle-\langle h_{\rm B} \;H_{\rm SB} \rangle + \langle h_{\rm B}^2\rangle $ can be evaluated analogously. We readily obtain $\langle H_{\rm SB}^2\rangle= t^2 \langle J_z^2 \rangle$, whereas $ \langle H_{\rm SB}\; h_{\rm B} \rangle= \langle h_{\rm B} \;H_{\rm SB} \rangle$
 can be found by following similar steps as in the $\langle h_{\rm B} \rangle$ calculation outlined above:
 $ \langle H_{\rm SB}\; h_{\rm B} \rangle=\zeta \,t^2 \langle J_z^2 \rangle $. Finally, we calculate $\langle h_{\rm B}^2\rangle$ as follows:
\begin{eqnarray*}
\langle h_{\rm B}^2 \rangle &= \frac{\zeta^2 t^2}{4\, \chi(t)} \sum_m |c_{mm}|^2 \,_{\rm B}\langle \sqrt{\chi(t)}\, m|   (a^{\dagger} a^{\dagger} +a^{\dagger} a + ( a^{\dagger}a + 1)  + a a | \sqrt{\chi(t)}\, m\rangle_{\rm B} \nonumber \\
&= \frac{\zeta^2 t^2}{4\,\chi(t)} \left[ 4 \,\langle J_z^2 \rangle\,\chi(t) + 1\right] . 
\end{eqnarray*}
In the first line we have used the commutation relationship $[a,a^{\dagger}]=1$ to write $a\, a^{\dagger}$ in terms of $a^{\dagger} a$, and then exploited the fact that coherent states are eigenvectors of $a$. This finally leads to
\begin{eqnarray*}
\langle \mathit{H}^2\rangle = t^2 \left(\langle J_z^2 \rangle  (1- \zeta)^2 + \zeta^2/ (4 \chi(t)) \right).
\end{eqnarray*}

\renewcommand{\theequation}{C.\arabic{equation}}
\setcounter{equation}{0}

\section{Additional technical details on 
controlled dephasing dynamics} 
\label{additional}

\subsection{Continuous control from limiting pulsed dynamics}
\label{app:continuous_limit}

We show here that continuous permutation-invariant driving can be obtained as a limit of the pulsed open-loop model in Eq.\,(\ref{Ht}). We consider an open-loop controlled dynamics governed by the noisy Hamiltonian 
\begin{equation}
H(t)=H_0(t)+\sum_{k=1}^Q u_k(t)\,\mathcal O_k, 
\label{eq:Hcont}
\end{equation}
where $H_0(t)$ is as in the main text and every $\mathcal O_k$ acts on the collective sector.
We partition the interval $[0,t]$ into $Q$ slices of duration $\Delta t=t/Q$, and define the pulse
$$V_j\equiv \exp\!\bigg(\!-i\sum_k u_k(t_j)\,\mathcal O_k\,\Delta t \bigg), $$
applied during slice $[t_j,t_{j+1}]$. The stroboscopic evolution operator is then
$$U_Q(t)=\prod_{j=Q-1}^{0}
\left[
V_j\,\exp\!\bigg\{-i\int_{t_j}^{t_{j+1}}\!\!\dd s\, H_0 (s)\bigg\}
\right].$$
Using the first-order Trotter formula,
$\exp[-i(A+B)\Delta t]
=
\exp(-iA\Delta t)\exp(-iB\Delta t)+O(\Delta t^2),$
we find
$$U_Q(t)=
\mathcal{T}\exp\!\bigg\{
-i\int_0^t \!\!\dd s\,
\Bigl(H_0(s)+\sum_{k=1}^K u_k(s)\, \mathcal O_k\Bigr)
\bigg\}
+\mathcal O(t\,\Delta t). $$
Thus, in the limit $Q\to\infty$ (i.e.,\ $\Delta t\to 0$, with fixed $t$), the pulsed evolution converges to the
continuously driven unitary generated by Eq.\,(\ref{eq:Hcont}).

\subsection{Proof of Lemma \ref{lem:compression_monotonicity}}
\label{lemmaproof}

\begin{proof}
Since $\Sigma \succ 0$, define the vector $\vec{w}\equiv \Sigma^{-1/2} \Delta \vec{t}$, such that $\Delta \vec{t}^{\,\top} \Sigma^{-1} \Delta \vec{t} = \vec{w}^\top \vec{w}.$
Moreover, we can rewrite the quadratic form in compressed notation
\begin{eqnarray*}
(\mathcal{S}\,\Delta \vec{t})^\top \widetilde \Sigma^{-1} (\mathcal{S}\, \Delta \vec{t})
&= \Delta \vec{t}^{\,\top} \mathcal{S}^\top (\mathcal{S}\,\Sigma \,\mathcal{S}^\top)^{-1} \mathcal{S} \Delta \vec{t} = \vec{w}^\top \Bigl[\Sigma^{1/2} \mathcal{S}^\top (\mathcal{S}\,\Sigma \,\mathcal{S}^\top)^{-1} \mathcal{S}\, \Sigma^{1/2}\Bigr] \vec{w}.
\end{eqnarray*}
Defining $P \equiv \Sigma^{1/2} \mathcal{S}^\top (\mathcal{S}\,\Sigma \,\mathcal{S}^\top)^{-1} \mathcal{S}\, \Sigma^{1/2},$
a direct computation shows that $P^\top=P$ and $P^2=P$, hence $P$ is an orthogonal projector.
Therefore, for all $\vec{w}$,
$$
0 \le \vec{w}^\top P \vec{w} \le \vec{w}^\top \vec{w}.$$
Substituting back yields Eq.(\ref{eq:compression_monotonicity}) in the main text.
\end{proof}

\subsection{Proof of Lemma \ref{lem:quadform}}
\label{lema1proof}

\begin{proof} We break the proof in four main steps.

\emph{Step 1: Structural decomposition of the covariance matrix.} 
We introduce the notation $t_i= a_i t $ for the times where rotations are applied, with $0< a_1<\ldots <a_Q<1$. We also define the $Q'$-dimensional finite-difference vectors (in bra/ket notation) $\{|\Delta a^{(\ell)}\rangle\}_{\ell=1}^{\infty}$, with components $|\Delta a^{\ell}\rangle_i= (a_{i+1}^{\ell}-a_i^{\ell})$. In this notation, the quadratic form is
\begin{equation*}
    \Delta \vec{t}^{\,\top} \Sigma^{-1} \Delta \vec{t} \equiv t^2 \bra{\Delta a^{(1)}}\Sigma^{-1} \ket{\Delta a^{(1)}}.
\end{equation*}
By Taylor expanding the entries of $\Sigma$, we find:
\begin{eqnarray*}
    \chi_{ij}(t) &= \frac{1}{4 \pi} \int_{-\infty}^{\infty}\!\dd\omega\,S(\omega) \int_{t_i}^{t_{i+1}}\!\!\!\dd s\int_{t_j}^{t_{j+1}} \!\!\dd s'\cos(\omega (s-s'))\nonumber\\
    &= \sum_{n=0}^{\infty} \frac{1}{2n!}\,(-1)^n  \frac{1}{4 \pi} \int_{-\infty}^{\infty}\!\dd\omega\,S(\omega)\,\omega^{2n}\left[ \int_{t a_i}^{t a_{i+1}}\!\!\!\dd s\int_{t a_j}^{t a_{j+1}} \!\!\dd s' (s-s')^{2n}\right] \nonumber\\
    &\equiv\sum_{n=0}^{\infty}\,(-1)^n\,\chi_0^n\, (\omega_c t)^{2 (n+1)} \Sigma_{ij}^{(n)}.
\end{eqnarray*}
The time dependence  enters as a power of the smallness parameter $\omega_c t$, whereas dimensionless constant $\chi_0^n$ captures the $n$-th spectral moment, and the geometric factor $\Sigma_{ij}^{(n)}$ depends exclusively on the arrangement of the pulses:
\begin{equation*}
 \omega_c^{2(n+1)}\chi_0^{n} \equiv \frac{1}{4 \pi} \int_{-\infty}^{\infty}\dd \omega\,S(\omega)\,\omega^{2n}, \quad   \Sigma_{i j}^{(n)} \equiv \frac{1}{2n!} \int_{ a_i}^{a_{i+1}}\!\!\!\dd u\int_{a_j}^{ a_{j+1}} \!\!\dd u' (u-u')^{2n} \nonumber.
\end{equation*}
The latter quantity can be  evaluated in terms of the finite difference vectors:
\begin{equation}
\begin{array}{rcl}
\Sigma_{ij}^{(n)}
&=&
\frac{1}{2n!}
\sum_{\ell = 0}^{2n}
\binom{2n}{\ell}
(-1)^{\ell}
\int_{a_i}^{a_{i+1}}\!\!\dd u
\int_{a_j}^{a_{j+1}}\!\!\dd u'\;
u^{\ell} {u'}^{2n-\ell}
\\[2ex]
&=&
\displaystyle
\sum_{\ell = 0}^{2n}
(-1)^{\ell}\;
\frac{|\Delta a^{(\ell+1)}\rangle_i}{(\ell+1)!}
\frac{|\Delta a^{(2n-\ell+1)}\rangle_i}{(2n-\ell+1)!}.
\end{array}
\end{equation}
Thus, the $n$-th geometric factor matrix $\Sigma^{(n)} \in M_{Q' \times Q'}$, with entries $\Sigma^{(n)}_{ij}$, is a sum of outer product contributions:
\begin{equation*}
    \Sigma^{(n)}= \sum_{\ell=1}^{2n+1} (-1)^{\ell-1}\frac{1}{\ell!}\,\frac{1}{(2(n+1)-\ell)!} |\Delta a^{(\ell)} \rangle \langle \Delta a^{(2(n+1)-\ell)} |.
\end{equation*}
We write, compactly, $\Sigma^{(n)} =B_{n} M^{(n)} B_{n}^{\top} $ with $B_n \equiv (|\Delta a^{(1)}\rangle, \ldots ,|\Delta a^{(2n+1)}\rangle)$ a $Q' \times (2n+1)$ matrix, and $M^{(n)}$ a $(2n+1) \times (2n+1)$ anti-diagonal matrix with entries
\begin{equation}
M^{(n)}_{i j}=\frac{1}{i!} \frac{1}{j!} (-1)^{(i-j)/2}\,\delta_{i+j,2(n+1)}, \quad i,j \in \{1,2n+1\} .
\label{antidiag}
\end{equation}
The lower-order contributions $\Sigma^{(m)}$ with $m <n$ can be embedded in the space spanned by a bigger set of difference vectors: $\Sigma^{(m)}= B_{n} M^{(m)} B_n^{\top}$. The matrix $M^{(m)}$ now must be interpreted as a square $(2n+1)$-dimensional  with a smaller populated antidiagonal:
\begin{equation*}
M^{(m)}_{ij}=\frac{1}{i!} \frac{1}{j!} (-1)^{(i-j)/2}\,\delta_{i+j,2(m+1)}, \quad \quad i,j \in \{1,2n+1\}   .
\end{equation*}
The covariance can be written exactly as the product of three infinite-dimensional matrices:
\begin{equation}
    \Sigma= \sum_{n=0}^{\infty} \chi_{0}^n\,(\omega_c t)^{2(n+1)}\, \Sigma^{(n)}  = B_{\infty} \mathcal{A}\, B_{\infty}^{\top} ,\quad  B_{\infty} \in M_{Q' \times \infty},\quad \mathcal{A} \in M_{\infty \times \infty}.   
    \label{dec}
\end{equation}
Here, $B_{\infty}= (|\Delta a^{(1)} \rangle, \ldots |\Delta a^{(\ell)}\rangle \ldots)$ contains all the difference vectors $\{|\Delta a^{(\ell)}\rangle \}_{\ell=0}^{\infty}$, and we have defined $ \mathcal{A} \equiv \sum_{n=0}^{\infty} \chi_0^{n}\,(\omega_c t)^{2(n+1)}\, M^{(n)}$.

\smallskip

\emph{Step 2: Short time expansion of $\langle \Delta a^{(1)}| \Sigma^{-1} |\Delta a^{(1)} \rangle$.}  Truncating Eq.\,(\ref{dec}) at finite order $R$ gives an approximation of the covariance matrix:
\begin{equation}
 \Sigma \approx \Sigma_R=B_R\, \mathcal{A}_R\,B_R^{\top},   \qquad \mathcal{A}_R \equiv \sum_{n=0}^R \chi_0^{n}\,(\omega_c t)^{2(n+1)}\, M^{(n)} . 
 \label{Ar}
\end{equation}
If $2R+1>Q'$ -- needed for a faithful approximation--, we can rearrange $\Sigma_R$  as the product of three \emph{square} matrices of dimension $Q'$\,:
\begin{equation*}
    \Sigma_{R}= B_{Q'}\, [C\mathcal{A}_R C^{\top}]\, B_{Q'}^{\top} =B_{Q'}\, \mathcal{J}_R^{(Q')}\, B_{Q'}^{\top}\,, \quad \quad C=B_{Q'}^{-1}B_{R},
\end{equation*}
with $\mathcal{B} \equiv \{|\Delta a^{(1)}\rangle, \ldots, |\Delta a^{(Q')}\rangle\}$ a basis of $\mathbb{R}^Q$.
 As $B_{Q'}^{-1} |\Delta a^{(\ell)}\rangle= e_{\ell}$  (the $\ell$th canonical vector) for $1 \leq \ell \leq Q'$, the first $Q'$ column vectors form the identity matrix and we can write compactly 
\begin{equation}
    C=( [I_{Q'}], B_{Q'}^{-1} |\Delta a^{(Q'+1)}\rangle, \ldots, B_{Q'}^{-1} |\Delta a^{(2R+1)}\rangle). 
    \label{mC}
\end{equation}
Inverting leads to the following short-time expansion of the quadratic form:
\begin{equation}
  t^2  \langle \Delta a^{(1)}| B_{Q'}^{-1} [\mathcal{J}_R^{(Q')}]^{-1} [B_{Q'}^{-1}]^{\top} |\Delta a^{(1)} \rangle = t^2 \,[\mathcal{J}_R^{(Q')}]^{-1}_{11}. 
  \label{bilinear}
\end{equation} 

\smallskip

\emph{Step 3: Decomposition of $[\mathcal{J}_R^{(Q')}]^{-1}_{11}$ using the Cauchy-Binet formula}.
To evaluate $[\mathcal{J}_R^{(Q')}]^{-1}_{11}$ we rewrite it in terms of the \emph{adjugate} matrix ${\rm adj}(\mathcal{J}_R^{(Q')})$:
\begin{equation} \!\!\!\!\!\!\!\!\!\!\!\!\!\!\!\!\!\!\!\!\!\!
[\mathcal{J}_R^{(Q')}]^{-1}_{11}= \frac{[{\rm adj}(\mathcal{J}_R^{(Q')})]_{11}}{{\rm det}([\mathcal{J}_R^{(Q')}])} = \frac{{\rm det}([\mathcal{J}_R^{(Q')}]_{1;1})}{{\rm det}(\mathcal{J}^{(Q')}_R)}, \qquad \mathcal{J}_{R}^{(Q')}\, {\rm adj}(\mathcal{J}^{(Q')}_{R})= {\rm det}(\mathcal{J}_R^{(Q')}),
\label{ar}
\end{equation}
with $[X]_{i;j}$ denoting the \emph{minor} matrix obtained after removing row $i$ and column $j$. 
We compute both the numerator and denominator in Eq.\,(\ref{ar}) using the Cauchy-Binet identity, which relates the determinant of a matrix product to the product of their minors:
\begin{equation*}
    {\rm det } (\mathcal{J}_R^{(Q')})=  {\rm det } (C\mathcal{A}_R C^{\top})= \langle v_C|\,G_{\mathcal{A}_R} |v_C\rangle.
\end{equation*}
Here, $|v_C\rangle$ is the column vector consisting of all the $Q'$-order minors of the matrix $C$, with components $v_{\ell}={\rm det}(C_{ \,;\mathcal{S}_{\ell}})$,  where $\mathcal{S}_{\ell}=\{ i_1,\ldots,i_{R-Q'}\} $ is a subset of $ \{1, \ldots 2R+1 \}$ of cardinality $|\mathcal{S}_{\ell}|=2R-Q'$  labeling the rows to be eliminated, and $1 \leq \ell \leq  {2R+1 \choose Q'} $. Meanwhile, the square matrix $G_{\mathcal{A}_R} $ contains all the $Q'$-order minors of matrix $\mathcal{A}_R$ as entries, obtained after eliminating $2R-Q'$ rows and columns: $[G_{\mathcal{A}_R}]_{\ell_1,\ell_2}={\rm det}([\mathcal{A}_R]_{ \mathcal{S}_{\ell_1};\mathcal{S}_{\ell_2}})$, with $S_{\ell_i}$ defined as above and $1 \leq \ell_i \leq {2R+1 \choose Q'} $, $i \in \{1,2\}$. 

\smallskip

\emph{Step 4: Computation of $[\tilde{\mathcal{J}}^{(Q')}_{R}]^{-1}_{11}$ short-time behavior.} 
We compute an arbitrary matrix element of $G_{\mathcal{A}_R}$
using the \emph{Leibniz formula}.
Let $S_{\ell_1}^{c} = \{r_1, \dots, r_{Q'}\}$ and
$S_{\ell_2}^{c} = \{c_1, \dots, c_{Q'}\}$ be subsets of
$\{1, \dots, 2R+1\}$ of cardinality $Q'$, obtained after removing from
$\mathcal{A}_R$ the row set $\mathcal{S}_{\ell_1}$ and column set
$\mathcal{S}_{\ell_2}$, respectively. We denote by
$\mathrm{Bij}(S_{\ell_1}^{c} \to S_{\ell_2}^{c})$ the set of all possible
$Q'!$ bijections between these sets. We write:
\begin{equation}
\det([\mathcal{A}_{R}]_{\mathcal{S}_{\ell_1}; \mathcal{S}_{\ell_2}})
=
\sum_{\sigma \in \mathrm{Bij}(S_{\ell_1}^{c} \to S_{\ell_2}^{c})}
\mathrm{sign}(\sigma)
\prod_{i \in S_{\ell_1}}
[\mathcal{A}_R]_{i,\sigma(i)}.
\label{detA}
\end{equation}
Here, $\mathrm{sign}(\sigma) \in \{ \pm 1\}$ is the of parity the bijection, which is $1$ for even permutations of $\{c_1,\ldots c_{Q'}\}$ and $-1$ otherwise. The matrix  $\mathcal{A}_{R}$ is sparse, and its entries are supported only on antidiagonals satisfying
$[\mathcal{A}_R]_{nm} \ne 0$ if and only if 
$n + m = 2(k+1)\ \mathrm{for\ some}\ 0 \leq k \leq R$.
Thus, only entries with even $n + m$ are nonzero.
A permutation $\sigma \in \mathrm{Bij}(S_1 \to S_2)$ gives a nonzero contribution 
$\prod_{i \in S_{\ell_1}} [\mathcal{A}_R]_{i, \sigma(i)}$ if and only if
$i + \sigma(i) \in 2\mathbb{Z} \ \mathrm{and}\ 
i + \sigma(i) \leq 2(R + 1),\ \forall i \in S_1$.
We denote such permutations as \textit{parity-compatible}. Each valid entry \( [\mathcal{A}_R]_{i, \sigma(i)} \) carries a factor of \( (\omega_c t)^{i + \sigma(i)} \), so the full permutation term contributes
\begin{equation*}
 \prod_{i \in S_1} [\mathcal{A}_R]_{i,\sigma(i)} \sim (\omega_c t)^{\sum_{i \in\mathcal{S}_{\ell_1}^{\complement} } (i + \sigma(i))}. 
\end{equation*}
Since  $\sigma$ is a bijection from $\mathcal{S}_{\ell_1}^{\complement} $ to \(\mathcal{S}_{\ell_2}^{\complement}  \), we have $\sum_{i \in \mathcal{S}_{\ell_1}^{\complement}} \sigma(i) = \sum_{j \in \mathcal{S}_{\ell_2}^{\complement}} j= \sum_{j=1}^{Q'} c_j$, and any parity-compatible \( \sigma \), has total exponent $E$ of \( \omega_c t \) 
\begin{equation}
E({\mathcal{S}^{\complement}_{\ell_1}; \mathcal{S}^{\complement}_{\ell_2}})=\sum_{i \in \mathcal{S}_{\ell_1}^{\complement} } (i + \sigma(i)) = \sum_{i=1}^{Q'} r_i + \sum_{j=1}^{Q'} c_j 
\label{exponent},
\end{equation}
independent of the specific permutation \( \sigma \). Thus, all parity-compatible terms in the Leibniz expansion contribute to the same leading order in $\omega_c t$. Any permutation with some \( i + \sigma(i) \) odd gives no contribution, and no lower powers of \( \omega_c t\) can appear:
\begin{equation}
\det([\mathcal{A}_{R}]_{\mathcal{S}_{\ell_1}; \mathcal{S}_{\ell_2}})= c_{\mathcal{S}_{\ell_1}; \mathcal{S}_{\ell_2}} \, (\omega_c t)^{E({\mathcal{S}_{\ell_1}; \mathcal{S}_{\ell_2}})}.
\label{det}
\end{equation}
The constant \( c_{\mathcal{S}_{\ell_1}; \mathcal{S}_{\ell_2}} \) can be expressed compactly in the form:
\begin{equation}
c_{\mathcal{S}_{\ell_1}; \mathcal{S}_{\ell_2}} = \sum_{ \sigma \in \mathcal{L}} \mathrm{sign}(\sigma) \prod_{i \in S^{\complement}_{\ell_1}}  \frac{1}{i! \, \sigma(i)!}\,(-1)^{(i- \sigma(i))/2}\, \chi_{\frac{i + \sigma(i)}{2} - 1} , 
\label{cs}
\end{equation}
with $\mathcal{L} \equiv \{\sigma \in \mathrm{Bij}(S_{\ell_1}^{c} \to S_{\ell_2}^{c}) :\,
i + \sigma(i)\ \mathrm{even}\ \forall i\}$ denoting the set of all parity-compatible bijections.
From Eq.\,(\ref{det}),  the leading order in the computation of Eq.\,(\ref{detA}) comes from the minor $\mathcal{S}_{\ell_1}^{\complement}=\mathcal{S}^{\complement}_{\ell_2}=\mathcal{S}_1^{\complement}=\{1,2,\ldots Q'\}$, which minimizes the exponent $E({\mathcal{S}^{\complement}_{\ell_1}; \mathcal{S}^{\complement}_{\ell_2}})$. Evaluating Eq.\,(\ref{exponent}) for this particular minor leads to $E(\mathcal{S}^{\complement}_1;\mathcal{S}^{\complement}_1)=Q' (Q'+1)$. This imposes a lower bound on $R$, as all even anti-diagonals of $[\mathcal{A}_{R}]_{\mathcal{S}_1,\mathcal{S}_1}$ must be populated if we wish to obtain the correct value of $c_{\mathcal{S}_1,\mathcal{S}_1}$. Thus, in the short-time limit:
\begin{equation*}
\det([\mathcal{J}_{R}^{(Q')}]) \approx v_{1}^2\,\;c_{\mathcal{S}_1;\mathcal{S}_1} (\omega_c t)^{Q' (Q'+1)} =\,c_{\mathcal{S}_1;\mathcal{S}_1} (\omega_c t)^{Q' (Q'+1)} \approx \det(\mathcal{A}_R), 
\end{equation*}
where the second equality is a consequence of $v_1=\det ([I_{Q'}])=1$.

Computation of ${\rm det}([\mathcal{J}^{(Q')}_R]_{1;1})$  breaks into similar steps, the only  difference being that $[\mathcal{J}^{(Q')}_R]_{1;1}= C_{;1} 
\mathcal{A}_R C_{;1}$ is now a square $Q$-dimensional matrix.
The leading-order minor now corresponds to partition $\mathcal{S}'_1=\{ 2,3, \ldots Q' \}$:
\begin{equation*}
 \det([\mathcal{A}_{R}]_{1;1}) \approx c_{\mathcal{S}'_1;\mathcal{S}'_1} (\omega_c t)^{Q' (Q'+1)-2},
\end{equation*}
for the modified constant
\begin{equation}
c_{\mathcal{S}'_1;\,\mathcal{S}'_1}
=
\sum_{\sigma \in \mathcal{L'}}
\mathrm{sign}(\sigma)(\sigma)
\prod_{i \in S'^{c}_{1}}
\frac{1}{i!\,\sigma(i)!}
(-1)^{(i-\sigma(i))/2}\;
\chi_{\frac{i + \sigma(i)}{2} - 1}.
\label{csp}
\end{equation}
and $\mathcal{L'}$ the set of parity compatible permutations of $\mathcal{S}'_1$. Replacing into Eq.\,(\ref{bilinear}), we find 
\begin{equation}
   t^2 \,[\mathcal{J}^{(Q')}_{R}]^{-1}_{11}  \approx t^2 \frac{\det([\mathcal{A}_R]_{1;1})}{\det(\mathcal{A}_R)}\approx t^2 \frac{  c_{\mathcal{S}_1;\mathcal{S}_1}  (\omega_c t)^{Q'(Q'+1)-2}}{  c_{\mathcal{S}'_1;\mathcal{S}'_1} (\omega_c t)^{Q'(Q'+1)}} = \omega_c^{-2}\, K_{Q'} ,  
   \label{KQdef}
\end{equation}
where $K_{Q'}=  \frac{  c_{\mathcal{S}_1;\mathcal{S}_1} }{  c_{\mathcal{S}'_1;\mathcal{S}'_1}}$  is \emph{independent} of any details of the partition chosen. 
\end{proof}

\subsection{Closed form for $K_{Q'}$ under Gaussian spectral cutoff}
\label{app:KQ_gaussian_corrected}

{We compute $K_{Q'}$ for the Gaussian-decaying spectrum in Eq.\,(\ref{Gausspec}), }
so that its moments are
\begin{equation}
\chi_0^n = \frac{1}{2\pi}\int_0^\infty S(\omega)\,\omega^{2n}\,d\omega=\frac{\alpha}{4\pi}\omega_c^2\,\Gamma\left(n+\frac{s+1}{2}\right)=A\,\Gamma(n+\gamma).
\label{eq:chi0_gauss_ABgamma}
\end{equation}


By Eq.\,(\ref{KQdef}), we write $K_{Q'}$ as the ratio of two matrix determinants,
 $   K_{Q'}= {\det \tilde{\mathcal{C}}}/{\det \mathcal{C}}.$
The matrices in question have vanishing entries whenever $i$ and $j$ have opposite parity 
\begin{equation}
\begin{array}{rcl}
[\mathcal{C}_{ij}]_{i,j=1}^{Q'}
&=&
\frac{1}{i!\,j!}\,(-1)^{(i-j)/2}\,
\chi_0^{i+j}\,
\delta_{i+j,\,2(n+1)},
\\[2ex]
[\tilde{\mathcal{C}}_{ij}]_{i,j=2}^{Q'}
&=&
\frac{1}{i!\,j!}\,(-1)^{(i-j)/2}\,
\chi_0^{i+j}\,
\delta_{i+j,\,2(n+1)}.
\end{array}
\label{mmC}
\end{equation}

We rearrange the indices via a permutation
$\pi_{Q'} : (1,2,3,4,5,6,\dots)
\mapsto
(1,3,5,\dots;\ 2,4,6,\dots)$.
The resulting matrix is block diagonal, regardless of the spectrum details:
\begin{equation}
\mathcal C
=
\left(
\begin{array}{cc}
\mathcal C_{\rm oo} & 0 \\[3pt]
0 & \mathcal C_{\rm ee}
\end{array}
\right),
\qquad
\det(\mathcal C)=\det(\mathcal C_{\rm oo})\det(\mathcal C_{\rm ee}).
\end{equation}
Subscripts $\mathrm{oo}$ and $\mathrm{ee}$ denote the odd-odd and even-even blocks, respectively. The sizes of the blocks depend on $Q'$:
\begin{equation}
\!\!\!\!\!\!\!\!\!\!\!\!\!\!\!\!\!\!\!\!\!\!\#\{\mathrm{odd\ indices\ in}\ 1\ldots Q'\}
=
\left\lceil \frac{Q'}{2} \right\rceil,
\qquad
\#\{\mathrm{even\ indices\ in}\ 1\ldots Q'\}
=
\left\lfloor \frac{Q'}{2} \right\rfloor.
\end{equation}
For definiteness, take $Q'=2n$.  
We label odd indices by $i=2a-1$ with $a=1,\dots,n$.  
Substituting into Eq.\,(\ref{mmC}) gives
\begin{equation}
(\mathcal C_{oo})_{ab}
=
\frac{(-1)^{a-b}}{(2a-1)!\,(2b-1)!}\;
\chi_0^{\,a+b-2},
\qquad a,b=1,\dots,n.
\label{eq:odd_block_explicit}
\end{equation}

We further introduce the diagonal matrices
$$
F^{(n)}_{\mathrm{odd}}
\equiv 
\mathrm{diag}\!\left(\frac{1}{(2a-1)!}\right)_{a=1}^n,
\qquad
S\equiv 
\mathrm{diag}\!\big((-1)^a\big)_{a=1}^n,
$$
and define the Hankel moment matrix
$H^{(n)}_{ab}
\equiv 
\chi_0^{\,a+b-2}.$
Then Eq.\,(\ref{eq:odd_block_explicit}) factorizes as a diagonal congruence
\begin{equation}
\mathcal C_{oo}
=
S\,F^{(n)}_{\mathrm{odd}}\,
H^{(n)}\,
F^{(n)}_{\mathrm{odd}}\,S.
\label{eq:odd_congruence}
\end{equation}
Taking determinants yields
\begin{equation}
\det(\mathcal C_{oo})
=
\det\!\big(F^{(n)}_{\mathrm{odd}}\big)^2\,
\det\!\big(H^{(n)}\big),
\label{eq:det_odd_factorized}
\end{equation}
since $\det(S)^2=1$. Similarly, we label even indices by $i=2a$ with $a=1,\dots,n$.  
An analogous computation for the even-even block gives:

$$
(\mathcal C_{ee})_{ab}
=
\frac{(-1)^{a-b}}{(2a)!\,(2b)!}\;
\chi_0^{\,a+b-1}=
S\,F^{(n)}_{\mathrm{even}}\,
\widetilde H^{(n)}\,
F^{(n)}_{\mathrm{even}}\,S,
\qquad a,b=1,\dots,n,
$$
where
$F^{(n)}_{\mathrm{even}}
\equiv 
\mathrm{diag}\!\left(\frac{1}{(2a)!}\right)_{a=1}^n,$ and $
\widetilde H^{(n)}_{ab}
\equiv 
\chi_0^{\,a+b-1}.$

We decompose ${\tilde{\mathcal C}}$ analogously. Deleting index $1$ removes one odd row and column but leaves the even sector unchanged.
The even–even determinants cancel exactly in the ratio defining $K_{Q'}$. 
The diagonal factorial factors in
Eq.\,(\ref{eq:det_odd_factorized})
cancel between numerator and denominator:
$$
\frac{
\det\!\big(F^{(n)}_{\mathrm{odd}}\big)^2
}{
\det\!\big(\tilde{F}^{(n-1)}_{\mathrm{odd}}\big)^2
}
=
1,
$$
as $\tilde{F}^{(n)}_{\mathrm{odd}}\equiv 
\mathrm{diag}\!\left(\frac{1}{(2a-1)!}\right)_{a=2}^n$.
Therefore $K_{Q'}$ reduces purely to a ratio of Hankel determinants built from the spectral moments,
$$
K_{2n} =
\frac{\det\!\big[\chi_0^{\,a+b}\big]_{a,b=1}^{n-1}
}{
\det\!\big[\chi_0^{\,a+b-2}\big]_{a,b=1}^{n}
}.
$$
The same expression holds for $K_{2n+1}$, as follows  from this parity structure.
We now specialize for the Gaussian-decaying spectrum by 
using Eq.\,(\ref{eq:chi0_gauss_ABgamma}),
\begin{eqnarray*}
\det[\chi_0^{\,a+b-2}]_{a,b=1}^{n}
&=
A^n\,\,
\det\!\big[\Gamma(\gamma+a+b)\big]_{a,b=0}^{n-1}, \\
\det[\chi_0^{\,a+b}]_{a,b=1}^{n-1}
&=
A^n\,\,
\det\!\big[\Gamma(\gamma+a+b+2)\big]_{a,b=0}^{n-2} .
\end{eqnarray*}
The matrices
\(
[\Gamma(\gamma+a+b)]_{a,b=0}^{n-1}
\) and \(
[\Gamma(\gamma+a+b+2)]_{a,b=0}^{n-2}
\)
are moment (Gram) matrices for the Laguerre weights
\(
x^{\gamma+1}e^{-x}
\) and $e^{-x}\,x^{\gamma+1}$, respectively.
A standard Gram determinant identity then gives the exact product formula
$$ 
\det\!\big[\Gamma(\gamma+a+b)\big]_{a,b=0}^{n-1}
=
\prod_{k=0}^{n-1} k!\,\Gamma(\gamma+k), \quad \det\!\big[\Gamma(\gamma+a+b+2)\big]_{a,b=0}^{n-2}
=
\prod_{k=0}^{n-2} k!\,\Gamma(\gamma+2+k). $$
The constant $K_Q$ is then given by the exact expression
$$
K_{2n}
=\frac{
A^{n-1}
\prod_{k=0}^{n-2} k!\,\Gamma(\gamma+2+k)
}{
A^{n}
\prod_{k=0}^{n-1} k!\,\Gamma(\gamma+k)
}=\frac{1}{A\,(n-1)!}\,
\frac{\Gamma(\gamma+n)}{\Gamma(\gamma)\Gamma(\gamma+1)}.
$$
Using $(n-1)!=\Gamma(n)$ and the Gamma ratio asymptotics
\(
\Gamma(n+\gamma)/\Gamma(n)
\sim n^{\gamma}
\)
as $n\to\infty$, we obtain
$$
K_{2n}
=
\frac{1}{A\,\Gamma(\gamma)\Gamma(\gamma+1)}
\,n^{\gamma}
\left(
1+\frac{\gamma(\gamma-1)}{2n}
+O(n^{-2})
\right).
$$

Finally, since $n\sim Q/2$, the large-$Q$ scaling may be obtained as 
$$
K_{Q'} \sim \mathrm{const}(s)\,Q^{\gamma}
=\mathrm{const}(s)\,Q^{(s+1)/2},
\qquad Q\to\infty.
$$
Thus, for a Gaussian-decaying spectrum, $K_{Q'}$ grows polynomially with exponent $(s+1)/2$. The dependence on $s$ enters both the exponent and the prefactor through Gamma functions.

\vspace{5mm}
\noindent\textbf{References}
\vspace{.3cm}

\bibliography{NoisyRamseyBib}
\bibliographystyle{bibstyle_v4}

\end{document}